\documentclass[lettersize,onecolumn]{IEEEtran}

\usepackage{amsmath,amssymb,amsthm}
\usepackage{graphicx}
\usepackage[numbers,sort&compress]{natbib}

\allowdisplaybreaks

\newcommand{\F}{\mathbb F}
\newcommand{\Prob}{\mathbb P}
\newcommand{\E}{\mathbb E}
\newcommand{\supp}{\operatorname{supp}}
\newcommand{\rank}{\operatorname{rank}}
\newcommand{\wt}{\operatorname{wt}}
\newcommand{\Core}{\operatorname{Core}}
\newcommand{\BP}{\mathrm{BP}}
\newcommand{\epsBP}{\epsilon_{\mathrm{BP}}}
\newcommand{\epsMAP}{\epsilon_{\mathrm{MAP}}}
\newcommand{\wex}{w_{\mathrm{ex}}}
\newcommand{\1}{\mathbf 1}

\newtheorem{theorem}{Theorem}[section]
\newtheorem{proposition}[theorem]{Proposition}
\newtheorem{lemma}[theorem]{Lemma}
\newtheorem{corollary}[theorem]{Corollary}

\title{
Limits of Weight-Bounded Redundant Checks for Peeling Decoding of Random Regular LDPC Codes
}

\author{Jianshen~Zhu
\thanks{Jianshen Zhu is with the Department of Information Sciences
and Technology, Tokyo University of Science, Noda, Chiba 278-8510,
Japan (e-mail: zhu\_js@rs.tus.ac.jp).}
\thanks{This work has been submitted to the IEEE for possible publication.
Copyright may be transferred without notice, after which this version may
no longer be accessible.}}

\begin{document}

\maketitle

\begin{abstract}
We study peeling decoding on the binary erasure channel when redundant
parity checks may be added without restriction in number, but each check
is subject to a row-weight constraint.  For random simple binary regular
low-density parity-check (LDPC) codes, we consider the parity-check matrix
formed by all nonzero dual codewords up to a prescribed weight.  We show
that, throughout an explicit logarithmic row-weight range, this complete
family of redundant checks has the same asymptotic normalized-residual
threshold as peeling with the original LDPC matrix.  In the same range,
with probability tending to one over the random code, it still has a
linear-size stopping set that is uniquely decodable.  Consequently, every
parity-check representation whose peeling decoder succeeds on every
uniquely decodable erasure pattern must contain checks above this
logarithmic scale, and the supports of those heavier checks must
collectively contain a linear number of coordinates.  Thus row weight
remains an asymptotic limitation even when the number and selection of
redundant checks are unrestricted.
\end{abstract}


\section{Introduction}
\label{sec:introduction}

On the binary erasure channel (BEC), the natural benchmark for block
decoding is unique recovery.  For an erasure set \(E\), the unknown
symbols satisfy a linear system obtained by restricting the parity-check
equations to the erased coordinates.  Recovery is unique exactly when this
restricted system has a unique solution, equivalently when the corresponding
columns of a parity-check matrix are linearly independent.  On the BEC,
this is precisely the condition for unique maximum-likelihood (ML) decoding.
The erased symbols can then be recovered by solving the linear system, for
example by Gaussian elimination.

A simpler alternative is belief-propagation (BP) decoding, which on the
BEC reduces to the peeling decoder~\cite{RichardsonUrbanke2001,RichardsonUrbanke2008}.
The peeling process repeatedly recovers an erased symbol whenever some parity
check contains exactly one currently erased coordinate.  For a bounded-degree
low-density parity-check (LDPC) code, it can be implemented in time linear in
the blocklength.  Its success, however, depends on the chosen parity-check
matrix.  A nonempty set of coordinates is a stopping set if every parity-check
row meets it in either zero or at least two positions.  The peeling decoder
fails to recover all erased symbols precisely when the erasure set contains
such a stopping set.  Hence, the decoder can fail even when the erasure
pattern is uniquely decodable~\cite{DiEtAl2002}.

A natural way to narrow the gap between ML and peeling decoding is to add
redundant parity checks, that is, additional dual codewords as rows of the
parity-check matrix.  This idea underlies stopping redundancy,
parity-check collections, and ML redundancy, which study how redundant
checks can remove stopping failures on correctable erasure
patterns~\cite{SchwartzVardy2006,HollmannTolhuizen2007,HanSiegel2008,
YakimenkaEtAl2019}.  A classical result gives a useful benchmark: every
binary \([n,k]\) code admits a parity-check matrix whose peeling decoder
succeeds on every uniquely decodable erasure pattern and whose nonzero rows
have weight at most \(k+1\)~\cite{HanSiegel2008,WeberAbdelGhaffar2008}.
For a fixed-rate LDPC family, this bound allows parity checks whose weight is
linear in the blocklength.
We study how far this row-weight requirement can be reduced
for random regular LDPC codes.  Rather than selecting a particular collection
of redundant checks, we give the peeling decoder every dual codeword up to a
prescribed weight \(W\).  This removes the number and choice of redundant
checks from consideration and isolates the effect of the row-weight limit
itself.

Given a binary linear code \(C\) and an integer \(W\geq1\), let
\(H_{\leq W}(C)\) denote the matrix whose rows are all nonzero dual
codewords of Hamming weight at most \(W\).  We call this matrix the
\emph{low-weight dual closure} at \(W\).  In the Tanner-graph formulation of
Bazzi and Audah~\cite{BazziAudah2015}, this is the graph
\(\overline{G}^{W}\) containing all redundant checks of degree at most
\(W\).  Thus \(H_{\leq W}(C)\) contains every parity check allowed by the
row-weight bound \(W\).
Figure~\ref{fig:correctable-erasure-failure} illustrates this maximal
choice at finite blocklengths for random simple \((3,6)\)- and
\((3,9)\)-regular LDPC codes.  The experiment is restricted to uniquely
decodable erasure patterns.  For each sampled code, the complete
low-weight dual closure is constructed exactly by exhaustive
enumeration of the full dual row space.  The low-weight dual closure
markedly lowers the peeling failure rate relative to the original
parity-check matrix, but does not eliminate such failures in the
displayed range.

\begin{figure*}[t]
    \centering
    \begin{minipage}[t]{0.485\textwidth}
        \centering
        \includegraphics[width=\linewidth]{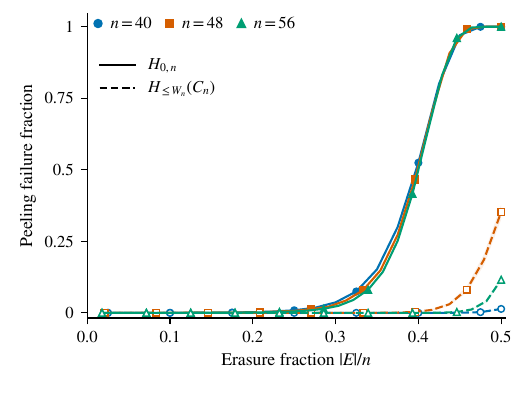}
        \par\smallskip
        {\small (a) $(3,6)$-regular ensemble}
    \end{minipage}
    \hfill
    \begin{minipage}[t]{0.485\textwidth}
        \centering
        \includegraphics[width=\linewidth]{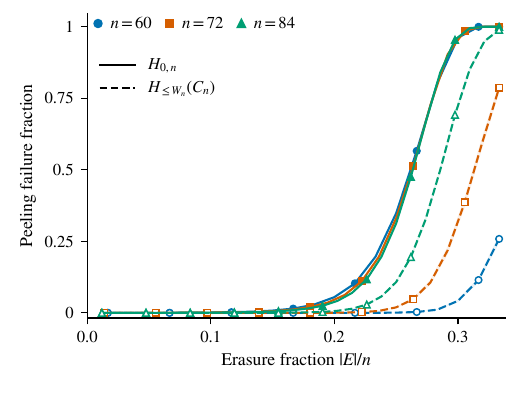}
        \par\smallskip
        {\small (b) $(3,9)$-regular ensemble}
    \end{minipage}
\caption{Peeling failure rates for uniquely decodable erasure patterns
in random simple regular LDPC codes.
The blocklengths are \(n=40,48,56\) for the \((3,6)\) ensemble and
\(n=60,72,84\) for the \((3,9)\) ensemble.
For the three blocklengths in each panel, respectively
\(150,150,200\) codes were sampled, with
\(500,500,1000\) accepted erasure sets per code and erasure size
\(t\); the corresponding ranges are
\(t=1,\ldots,20\), \(1,\ldots,24\), and \(1,\ldots,28\).
For each code and \(t\), \(E\) is sampled uniformly among the
\(t\)-subsets satisfying
\(\operatorname{rank}_{\F_2}((H_{0,n})_E)=t\).
The same accepted erasure sets are evaluated using the original matrix
\(H_{0,n}\) and the complete low-weight dual closure
\(H_{\leq W_n}(C_n)\), which is obtained by exhaustive enumeration of
the full dual row space.
The plotted value is the equal-weight mean of the code-level peeling
failure fractions.
Solid curves correspond to \(H_{0,n}\), and dashed curves to
\(H_{\leq W_n}(C_n)\), with
\(W_n=\lfloor0.75c_*\log n\rfloor\).
The nominal cutoffs are \(9,9,10\) for \((3,6)\) and \(15,16,16\)
for \((3,9)\); because every dual word in the \((3,6)\) cases has
even weight, the first two effective cutoffs are \(8\).
Pointwise \(95\%\) bootstrap intervals for the ensemble-averaged
failure fraction, obtained from \(5000\) whole-code resamples, are
also shown.  They quantify the finite-code-sample uncertainty in the
equal-weight code-level mean; the intervals are narrow on this scale
and therefore mostly coincide visually with the corresponding curves.
}
    \label{fig:correctable-erasure-failure}
\end{figure*}

In this paper, we focus on random regular LDPC codes and identify an
explicit logarithmic range in which low-weight redundant checks are
insufficient to close the gap between peeling and unique decoding.
More precisely, for a uniformly random simple \((\ell,r)\)-regular code
with \(3\leq\ell<r\), we show that there is an explicit constant
\(c_*=c_*(\ell,r)>0\) such that, for every fixed \(\gamma>0\), the
low-weight dual closure up to weight \((c_*-\gamma)\log n\) has the same
asymptotic normalized-residual threshold as peeling with the original
LDPC matrix.  Thus, throughout this logarithmic range, making every
admissible low-weight dual check available does not improve the BEC
peeling threshold.

We also show that the same weight range is insufficient for exact recovery
of all uniquely decodable erasure patterns.  With probability tending to one
over the random code, the low-weight dual closure contains a linear-size
stopping set that is nevertheless uniquely decodable.  Consequently, any
parity-check matrix whose peeling decoder succeeds on every uniquely
decodable erasure pattern must contain a row of weight greater than
\((c_*-\gamma)\log n\).  Moreover, the supports of the checks above this
weight scale must collectively contain a linear number of coordinates.

The closest asymptotic comparison is the work of Bazzi and
Audah~\cite{BazziAudah2015}, who study the same complete family of
redundant checks subject to a check-degree cutoff in the setting of LP
decoding.  Their results show that, under structural assumptions on the Tanner
graph, the LP threshold of a sufficiently large bounded-degree
redundant-check graph can be made arbitrarily close to that of the
full redundant-check graph.  Under an additional rigidity assumption,
the full redundant-check graph has the same LP threshold as the base
graph on the binary symmetric channel.  Our setting instead concerns
peeling on the BEC for the random regular LDPC ensemble.  We identify an
explicit logarithmic row-weight range in which the complete low-weight
dual closure leaves the normalized-residual threshold unchanged and still
fails on some uniquely decodable erasure patterns.

Related BEC work has emphasized different redundancy resources.
Wadayama studies redundant random matrix ensembles and the tradeoff between
the number of redundant rows and their stopping-set distributions~\cite{Wadayama2008}.
At finite lengths, Bocharova~et~al. augment selected short LDPC
parity-check matrices with low-weight dual codewords and study the
resulting gap between BP and ML decoding~\cite{BocharovaEtAl2017}.
In the present work, the decoder is instead
given every dual word below the prescribed weight cutoff.  This removes
the choice and number of redundant rows as design parameters and isolates
the effect of the row-weight constraint itself.

The remainder of the paper is organized as follows.
Section~\ref{sec:model} introduces the model and basic definitions.
Section~\ref{sec:barrier} states the main results.
Section~\ref{sec:proof-locality-barrier} gives their proofs, based on
structural properties of low-weight dual checks and stability of the
residual produced by peeling with the original parity-check matrix.
Section~\ref{sec:discussion} concludes the paper.
The technical estimates used in the proofs are collected in the appendix.

\section{Redundant parity checks under a row-weight constraint}
\label{sec:model}

Fix integers $3\leq \ell<r$.  We work with the uniform simple
$(\ell,r)$-regular Tanner ensemble.  For each blocklength $n$ for
which $m=\ell n/r$ is an integer, let
$G_n=(V_n,F_n,E_n)$ be a graph drawn from this ensemble, with
$|V_n|=n$ and $|F_n|=m$.  We identify $V_n$ with $[n]$ for convenience.
The binary incidence matrix of $G_n$ is denoted by $H_{0,n}$, and
$h_f$ denotes its row indexed by $f\in F_n$.  The corresponding
binary linear code is
\(
    C_n=\ker H_{0,n}.
\)

For a fixed parity-check matrix, the failure patterns of the peeling
decoder on the BEC are completely characterized by its stopping
sets~\cite{DiEtAl2002}.  A set $S\subseteq[n]$ is a
\emph{stopping set} for a binary matrix $H$ if no row of $H$
contains exactly one coordinate of $S$.
Starting from an erasure
set $E$, the peeling process repeatedly recovers an erased coordinate
whenever some check contains exactly one currently erased
coordinate.  When no such check remains, the unrecovered coordinates
form the largest stopping set contained in $E$, which we denote by
\(
    \Core_H(E).
\)
In particular, for the matrix $H_{0,n}$ we abbreviate this set as
\(
    R_0(E):=\Core_{H_{0,n}}(E).
\)
We refer to $R_0(E)$ as the residual set after peeling with
$H_{0,n}$; it consists of the erased coordinates left unrecovered
when the process stops.

For $\epsilon\in[0,1]$, let $E_\epsilon\subseteq[n]$ be the random
erasure set in which each coordinate is erased independently with
probability $\epsilon$, independently of $G_n$.
Unless otherwise stated, probabilities and expectations are taken
over the uniform simple Tanner graph $G_n$ and, whenever
$E_\epsilon$ is present, over the independent BEC erasures as well.
The asymptotic behavior of
the peeling decoder for the $(\ell,r)$-regular ensemble is described
by density evolution~\cite{RichardsonUrbanke2001,DiEtAl2002,RichardsonUrbanke2008}.
In the density-evolution recursion, let $x_t$ denote the erasure
probability of a variable-to-check message at iteration $t$, and let
$y_t$ denote the erasure probability of a check-to-variable message.
For the $(\ell,r)$-regular ensemble,
\[
    y_t=1-(1-x_t)^{r-1},
    \qquad
    x_{t+1}=\epsilon y_t^{\ell-1},
\]
with $x_0=\epsilon$.  The corresponding asymptotic fraction of
variables that remain erased after iteration $t$ is
$\epsilon y_t^\ell$.  Hence the nonzero fixed points of the recursion
satisfy
\begin{equation}
    x=\epsilon y^{\ell-1},
    \qquad
    y=1-(1-x)^{r-1}.
    \label{eq:de}
\end{equation}
Equivalently, define
\[
    f_\epsilon(z)
    :=
    \epsilon\bigl[1-(1-z)^{r-1}\bigr]^{\ell-1}.
\]
Then the fixed-point equation is simply $x=f_\epsilon(x)$.

We denote by $\epsBP=\epsBP(\ell,r)$ the belief-propagation (BP)
threshold of the $(\ell,r)$-regular ensemble on the BEC, that is,
the supremum of the erasure probabilities $\epsilon$ for which the
density-evolution recursion, initialized at $x_0=\epsilon$, converges
to zero~\cite{RichardsonUrbanke2001,RichardsonUrbanke2008}.  
For \(3\leq\ell<r\), this is equivalently the value at which a
nonzero fixed point first appears by tangency.  For every fixed
\(\epsilon>\epsBP\), there are exactly two positive fixed points, and
the larger one is the stable upper branch~\cite[Lemmas~4 and~9]{KudekarRichardsonUrbanke2011}.
By monotonicity of the scalar density-evolution recursion, initialization
at \(x_0=\epsilon\) converges to this larger fixed point.  If \(x\)
denotes it, then
\[
    f_\epsilon'(x)<1.
\]
The upper branch approaches the BP tangency continuously as
\(\epsilon\downarrow\epsBP\), as follows from the continuous
fixed-point parametrization and the tangency characterization above.
Standard concentration results then imply that, for every fixed
$\epsilon<\epsBP$,
\[
    \frac{|R_0(E_\epsilon)|}{n}
    \xrightarrow{\Prob}0,
\]
whereas for every fixed $\epsilon>\epsBP$ the fraction of variables
left unrecovered by the peeling process is bounded away from zero
with high probability~\cite{DiEtAl2002,RichardsonUrbanke2008}.

Stopping sets depend on the chosen parity-check matrix, not only on
the code itself.  Adding redundant parity checks can therefore
change the behavior of the peeling decoder without changing the
code.
For a binary linear code $C\leq\F_2^n$ and an integer $W\geq1$, let
$H_{\leq W}(C)$ denote the matrix whose rows are all nonzero
$h\in C^\perp$ with $\wt(h)\leq W$.
We call this matrix the \emph{low-weight dual closure}.
For every $f\in F_n$, the row $h_f$ belongs to $C_n^\perp$ and has
weight $r$.  Hence, for $W\geq r$, $H_{\leq W}(C_n)$ contains every row of
$H_{0,n}$, while every row of $H_{\leq W}(C_n)$ is a parity check
of $C_n$.  Therefore
\(
    \ker H_{\leq W}(C_n)=C_n.
\)

To study whether the added low-weight checks can change the BEC
threshold, let $W_\bullet=(W_n)$ be a deterministic sequence of
positive integers with $W_n\geq r$ for all sufficiently large $n$,
and define
\begin{equation}
    \epsilon_{\mathrm{res}}(W_\bullet)
    :=
    \sup\left\{
        \epsilon\in[0,1]:
        \frac{
            \left|
            \Core_{H_{\leq W_n}(C_n)}(E_\epsilon)
            \right|
        }{n}
        \xrightarrow{\Prob}0
    \right\}.
    \label{eq:residual-threshold}
\end{equation}
Notice that the same definition with $H_{0,n}$ in
place of $H_{\leq W_n}(C_n)$ gives $\epsBP$, by the concentration
result recalled above.

On the other hand, the peeling decoder associated with a given
parity-check matrix need not agree with maximum-likelihood (ML)
decoding on every erasure pattern: it may stop on a stopping set even
when the erased symbols are uniquely determined by the code~\cite{HanSiegel2008,YakimenkaEtAl2019}.
To formalize this
distinction, let $H$ be a parity-check matrix for a binary linear
code $C$ and let $E\subseteq[n]$.  Denote by $H_E$ the submatrix
formed by the columns indexed by $E$.  The erased symbols are
uniquely determined from the unerased ones if and only if
\(
    \rank(H_E)=|E|.
\)
An erasure set satisfying this condition is called
\emph{coverable}~\cite{HanSiegel2008,YakimenkaEtAl2019}.  
Since \(\operatorname{rowspan}H=C^\perp\), coverability depends only
on \(C\) and \(E\), whereas the stopping-set property depends on the
chosen parity-check representation.
A stopping set $S$ of $H$ is called a \emph{coverable stopping set} if it is
coverable as an erasure set.  We say that $H$ is
\emph{ML-exact on the BEC} if the peeling decoder based on $H$
succeeds for every coverable erasure set.
The following criterion characterizes ML-exactness in terms of
stopping sets.

\begin{proposition}
\label{prop:rank-interface}
Let $H$ be a parity-check matrix for a binary linear code $C$.
Then $H$ is ML-exact on the BEC if and only if
\(
    \rank(H_S)<|S|
\)
for every nonempty stopping set $S$ of $H$.
\end{proposition}

\begin{proof}
If a nonempty stopping set $S$ satisfies
$\rank(H_S)=|S|$, then $S$ is coverable, while the peeling decoder
based on $H$ cannot recover any coordinate when the erased set is
$S$.  Hence $H$ is not ML-exact on the BEC.

Conversely, suppose that $H$ is not ML-exact on the BEC.  Then there
is a coverable erasure set $E$ for which the peeling decoder based on
$H$ does not recover all erased coordinates.  The set
\(
    S=\Core_H(E)
\)
is therefore a nonempty stopping set contained in $E$.  Since
$\rank(H_E)=|E|$, the columns of $H_E$ are linearly independent.
The columns indexed by $S$ are consequently linearly independent as
well, and hence
\(
    \rank(H_S)=|S|.
\)
\end{proof}

We measure the row weight required for ML-exactness by
\begin{equation}
    \wex(C)
    :=
    \min_H
    \max\{
        \wt(h):
        h\text{ is a nonzero row of }H
    \},
    \label{eq:wex-definition}
\end{equation}
where the minimum is over all parity-check matrices $H$ for $C$
that are ML-exact on the BEC.  If $H$ has no nonzero rows, the inner
maximum is defined to be zero.

For a fixed row-weight budget, it is enough to consider the
low-weight dual closure.  The following proposition gives the
precise reduction.

\begin{proposition}
\label{prop:budget-domination}
Let $H$ be a parity-check matrix for a binary linear code $C$, and
suppose that every nonzero row of $H$ has weight at most $W$.
Then $H_{\leq W}(C)$ is also a parity-check matrix for $C$.
Moreover, for every erasure set $E$,
\(
    \Core_{H_{\leq W}(C)}(E)
    \subseteq
    \Core_H(E),
\)
and for every $S\subseteq[n]$,
\(
    \operatorname{rowspan}(H_S)
    =
    \operatorname{rowspan}\bigl((H_{\leq W}(C))_S\bigr).
\)
In particular,
\[
    \rank(H_S)
    =
    \rank\bigl((H_{\leq W}(C))_S\bigr).
\]
\end{proposition}

\begin{proof}
Since every nonzero row of $H$ has weight at most $W$,
\[
    C^\perp
    =
    \operatorname{rowspan} H
    \subseteq
    \operatorname{rowspan} H_{\leq W}(C)
    \subseteq
    C^\perp.
\]
Thus $H_{\leq W}(C)$ is a parity-check matrix for $C$.

Every stopping set of $H_{\leq W}(C)$ is also a stopping set of $H$,
because the former matrix contains every nonzero row of the latter.
Taking the largest stopping set contained in $E$ gives
\(
    \Core_{H_{\leq W}(C)}(E)
    \subseteq
    \Core_H(E).
\)

For $S\subseteq[n]$, let
$\pi_S:\F_2^n\to\F_2^S$ denote restriction to the coordinates in
$S$.  Both $H$ and $H_{\leq W}(C)$ have row space $C^\perp$, and
therefore
\[
    \operatorname{rowspan}(H_S)
    =
    \pi_S(C^\perp)
    =
    \operatorname{rowspan}\bigl((H_{\leq W}(C))_S\bigr).
\]
The rank equality follows.
\end{proof}

Thus \(H_{\leq W}(C)\) leaves no larger peeling residual than any
parity-check matrix for \(C\) whose nonzero rows have weight at most
\(W\).  In particular, if \(H_{\leq W}(C)\) has a coverable stopping
set, then no such parity-check matrix can be ML-exact on the BEC.

\section{Main results}
\label{sec:barrier}

We first state the main results in this section.  Their proofs are
given in Section~\ref{sec:proof-locality-barrier},
with supporting technical estimates proved in the appendix.

Let $(x_{\BP},y_{\BP})$ denote the positive solution of
\eqref{eq:de} at $\epsilon=\epsBP$ satisfying
\(
    f_{\epsBP}'(x_{\BP})=1.
\)
Throughout, $\log$ denotes the natural logarithm.  Set
\[
    q:=r-2,\qquad
    b:=(\ell-1)(r-1),\qquad
    c_0:=\frac{q}{\log b},
\]
and define
\[
    c_{\mathrm{gate},\BP}
    :=\frac{1}{-\log(1-x_{\BP})},
    \qquad
    c_*:=\min\{c_{\mathrm{gate},\BP},2c_0\}.
\]

Our first theorem shows that, for every fixed margin below the scale
$c_*\log n$, adding all dual checks within the row-weight budget does
not change the normalized-residual threshold from $\epsBP$.

\begin{theorem}
\label{thm:threshold}
Fix $3\leq\ell<r$.  For every fixed $\gamma>0$ and every deterministic sequence of
positive integers $W_\bullet=(W_n)$ eventually satisfying
\[
    r\leq W_n\leq(c_*-\gamma)\log n,
\]
we have
\[
    \epsilon_{\mathrm{res}}(W_\bullet)
    =
    \epsBP(\ell,r).
\]
\end{theorem}

The same row-weight range also precludes ML-exactness.  In fact, with high
probability over $G_n$, the low-weight dual closure has
a coverable stopping set of linear size.

\begin{theorem}
\label{thm:coverable}
Fix $3\leq\ell<r$ and $\gamma>0$.  There exists
$\eta=\eta(\ell,r,\gamma)>0$ such that, for every deterministic
sequence of positive integers $W_\bullet=(W_n)$ eventually satisfying
\[
    r\leq W_n\leq(c_*-\gamma)\log n,
\]
with probability $1-o(1)$ over $G_n$, the matrix
$H_{\leq W_n}(C_n)$ has a coverable stopping set $S_n$ with
\[
    |S_n|\geq\eta n.
\]
\end{theorem}

Theorem~\ref{thm:coverable}, together with
Propositions~\ref{prop:rank-interface} and
\ref{prop:budget-domination}, gives a lower bound on the maximum row
weight required for ML-exactness.

\begin{corollary}
\label{cor:wex}
For every fixed $\gamma>0$,
\[
    \Prob\!\left\{
        \wex(C_n)>(c_*-\gamma)\log n
    \right\}
    \longrightarrow1,
\]
where the probability is over $G_n$.
\end{corollary}

\begin{proof}
Suppose first that $0<\gamma<c_*$ and set
\(
    W_n:=\left\lfloor(c_*-\gamma)\log n\right\rfloor.
\)
For all sufficiently large $n$, we have $W_n\geq r$.  By
Theorem~\ref{thm:coverable}, with probability $1-o(1)$ the matrix
$H_{\leq W_n}(C_n)$ has a coverable stopping set.  On this event,
Propositions~\ref{prop:budget-domination} and
\ref{prop:rank-interface} rule out every ML-exact parity-check matrix
whose nonzero rows have weight at most $W_n$.  Thus
\(
    \wex(C_n)>W_n=\lfloor(c_*-\gamma)\log n\rfloor,
\)
which is equivalent to the asserted bound because $\wex(C_n)$ is
integer valued.
For $\gamma\geq c_*$, the claim is immediate.
\end{proof}

Thus, for every fixed $\gamma>0$, with probability tending to one
over $G_n$, every ML-exact parity-check matrix for $C_n$ contains a
row of weight exceeding $(c_*-\gamma)\log n$.

Corollary~\ref{cor:wex} only requires one row above the logarithmic
weight scale.  The following result shows that, for ML-exactness,
such heavier rows must collectively meet a linear number of
coordinates.

\begin{theorem}
\label{thm:heavy-support}
Fix $3\leq\ell<r$ and $\gamma>0$, and set
\(
    a_n:=(c_*-\gamma)\log n.
\)
For a parity-check matrix $H$ for $C_n$, let
\[
    U_{>a_n}(H)
    :=
    \bigcup_{\substack{h\text{ row of }H\\ \wt(h)>a_n}}
    \supp(h).
\]
There exists $\delta=\delta(\ell,r,\gamma)>0$ such that, with
probability $1-o(1)$ over $G_n$, every ML-exact parity-check matrix
$H$ for $C_n$ satisfies
\[
    |U_{>a_n}(H)|\geq\delta n.
\]
\end{theorem}

Since
\[
    |U_{>a_n}(H)|
    \leq
    \sum_{\substack{h\text{ row of }H\\ \wt(h)>a_n}}
    \wt(h),
\]
the same high-probability event gives
\[
    \sum_{\substack{h\text{ row of }H\\ \wt(h)>a_n}}
    \wt(h)
    \geq
    \delta n
\]
for every ML-exact parity-check matrix $H$.
This incidence bound also gives a lower bound on the number of
heavier rows when their individual weights are bounded.  If every
nonzero row of $H$ has weight at most $L_n$, then
\[
    \#\bigl\{
        h\text{ row of }H:\wt(h)>a_n
    \bigr\}
    \geq
    \frac{\delta n}{L_n}.
\]
These conclusions continue to hold after zero rows and duplicate
copies are removed.

\section{Proofs of the main results}
\label{sec:proof-locality-barrier}

We now prove the three main results stated in
Section~\ref{sec:barrier}.  The common step is to construct a large
stopping set for the low-weight dual closure
$H_{\leq W_n}(C_n)$.

For $\epsilon>\epsBP$, the residual set $R_0(E_\epsilon)$ left by
the original matrix $H_{0,n}$ has linear size with high probability.
It is a stopping set of $H_{0,n}$, but it need not remain a stopping
set after redundant parity checks are added.
Section~\ref{sec:proof-locality-barrier}-A identifies a set of coordinates of vanishing normalized
size with the property that any exact-one intersection caused by a
logarithmic-weight dual row yields another exact-one intersection
supported within that set.  Section~\ref{sec:proof-locality-barrier}-B proves a separate stability
result: after deleting a sufficiently small subset of
$R_0(E_\epsilon)$, restarting the peeling process based on
$H_{0,n}$ still leaves a stopping set of linear size.
Section~\ref{sec:proof-locality-barrier}-C combines these two ingredients to prove the main
results.

\subsection{Localization and cyclic reduction of low-weight dual checks}

Because $H_{\leq W}(C_n)$ contains arbitrary low-weight dual words,
we first relate each such word to a small collection of the original
checks.  For $J\subseteq F_n$, define
\(
    h_J:=\sum_{f\in J} h_f \in \F_2^n,
\)
where the sum is taken over $\F_2$.  Thus $h_J$ is the redundant
parity check obtained by adding the original check rows indexed by
$J$.  Equivalently, a coordinate belongs to $\supp(h_J)$ if it is
contained in an odd number of checks in $J$.
For nonzero $h\in C_n^\perp$, define
\[
    \sigma(h)
    :=
    \min\{
        |J|:J\subseteq F_n,\ h_J=h
    \}.
\]
Thus $\sigma(h)$ is the smallest number of original check rows whose
binary sum equals $h$.
Notice that a small value of $\wt(h)$ does not by itself imply that
$\sigma(h)$ is small, since cancellations among many rows of
$H_{0,n}$ can reduce the weight of their sum.  What we need below is
the existence of a small generating set for every low-weight dual
word.  The following theorem gives this localization.

\begin{theorem}
\label{thm:localization}
Fix $A<\infty$ and $3\leq\ell<r$.  For every deterministic
integer-valued sequence $W_n\leq A\log n$, with probability
$1-o(1)$ over $G_n$, every nonzero $h\in C_n^\perp$ with
$\wt(h)\leq W_n$ satisfies
\(
    \sigma(h)\leq W_n.
\)
\end{theorem}

The proof is given in Appendix~\ref{app:dual-checks}.  Thus, with
probability $1-o(1)$ over $G_n$, every nonzero
$h\in C_n^\perp$ with $\wt(h)\leq W_n$ admits a set
$J\subseteq F_n$ such that
\(
    h=h_J,
    |J|\leq W_n.
\)

We next examine how such a generating set can produce an exact-one
intersection with a stopping set of $H_{0,n}$.  For  $J\subseteq F_n$,
let $T(J)$ denote the subgraph of $G_n$
consisting of the checks in $J$, all variables adjacent to them, and
all incident edges from the checks in $J$.
The connected components of $T(J)$ induce a partition
\(
    J
    =
    J_1\mathbin{\dot\cup}\cdots\mathbin{\dot\cup}J_k
\)
of $J$, where each $T(J_i)$ is a connected component
of $T(J)$.  In particular, the graphs $T(J_i)$ have disjoint variable
sets, and hence
\(
    \supp(h_J)
    =
    \biguplus_{i=1}^k \supp(h_{J_i}).
\)
Suppose now that $S$ is a stopping set of $H_{0,n}$ and
\(
    |S\cap\supp(h_J)|=1.
\)
Then there is a unique index $i$ for which
\(
    |S\cap\supp(h_{J_i})|=1.
\)
Denote the corresponding check set by $J_0$.  Replacing $J$ by
$J_0$ preserves the exact-one intersection and can only reduce the
number of checks and the weight of the generated row.
The graph $T(J_0)$ is connected and must contain a cycle.  Indeed,
Lemma~\ref{lem:tree-inert} in Appendix~\ref{app:dual-checks} rules out
the tree case: if $T(J_0)$ were a tree, then $h_{J_0}$ could not
meet a stopping set of $H_{0,n}$ in exactly one coordinate.

For $K\subseteq F_n$ and a variable $v\in V_n$, define
\(
    d_K(v):=|N(v)\cap K|.
\)
Let $B(K)$ be the bipartite graph with check-node set $K$,
variable-node set
\(
    \{v\in V_n:d_K(v)\geq2\},
\)
and all edges of $G_n$ between these nodes.  Thus $B(K)$ records
the variables shared by at least two checks in $K$, while variables
incident to only one check in $K$ are omitted.  We call $K$
a \emph{connected cyclic skeleton} if $B(K)$ is
nonempty and connected and every check node and every variable
node of $B(K)$ has degree at least two.
Proposition~\ref{prop:stopping-active-reduction} in
Appendix~\ref{app:dual-checks} shows that the exact-one component $J_0$
can be reduced to such a skeleton without increasing either the
number of checks or the weight of the generated dual word.
More precisely, there exists $K\subseteq J_0$ such that
\[
    |S\cap\supp(h_K)|=1,
    \qquad
    |K|\leq |J_0|\leq |J|,
    \qquad
    \wt(h_K)\leq\wt(h_{J_0})\leq\wt(h_J),
\]
and $K$ is a connected cyclic skeleton.

For a connected cyclic skeleton $K$, let
\(
    \mu(K):=|E(B(K))|-|V(B(K))|+1
\)
be the cyclomatic number of $B(K)$.  Theorem~\ref{thm:skeleton-identity}
in Appendix~\ref{app:dual-checks} shows that
\(
    |K|\leq(\wt(h_K)+2\mu(K)-2)/q.
\)
We treat separately the cases $\mu(K)=1$ and $\mu(K)\geq2$.

Suppose first that $\mu(K)=1$.  By
Theorem~\ref{thm:skeleton-identity}, $B(K)$ is an alternating cycle
in the simple Tanner graph.  Write its check nodes as
$f_1,\ldots,f_s$ and its variable nodes as
$v_1,\ldots,v_s$, labeled cyclically so that $f_i$ is adjacent to
$v_{i-1}$ and $v_i$, with indices taken modulo $s$.
For each $i$, let
\(
    A_i:=N(f_i)\setminus\{v_{i-1},v_i\}.
\)
Since $B(K)$ contains every variable adjacent to at least two checks
in $K$, and here $B(K)$ consists exactly of the cycle above, every
variable in $A_1\cup\cdots\cup A_s$ is adjacent to exactly one check
in $K$.  Hence the sets $A_1,\ldots,A_s$ are pairwise disjoint and
are disjoint from $\{v_1,\ldots,v_s\}$.
In the binary sum
\(
    h_K=\sum_{i=1}^s h_{f_i},
\)
each variable $v_i$ appears twice and cancels, whereas every variable
in $A_1\cup\cdots\cup A_s$ appears once.  Therefore
\[
    \supp(h_K)
    =
    A_1\mathbin{\dot\cup}\cdots\mathbin{\dot\cup}A_s,
    \qquad
    \wt(h_K)=qs.
\]

For the later count, we need one further consequence of the
stopping-set condition for these $\mu(K)=1$ configurations.

\begin{lemma}
\label{lem:cycle-gate}
Let $K$ be a connected cyclic skeleton with $\mu(K)=1$, with
$f_1,\ldots,f_s$, $v_1,\ldots,v_s$, and $A_1,\ldots,A_s$ as above.
If $S$ is a stopping set of $H_{0,n}$ and
\(
    |S\cap\supp(h_K)|=1,
\)
then
\(
    v_1,\ldots,v_s\in S.
\)
\end{lemma}

\begin{proof}
Let $u$ be the unique element of $S\cap\supp(h_K)$.  Since
\(
    \supp(h_K)
    =
    A_1\mathbin{\dot\cup}\cdots\mathbin{\dot\cup}A_s,
\)
there is a unique index $j$ such that $u\in A_j$.
Thus $A_j\cap S=\{u\}$ and $A_i\cap S=\varnothing$ for every $i\neq j$.

Fix $i\neq j$.  None of the variables in $A_i$ belongs to $S$, so
the only neighbors of $f_i$ that can belong to $S$ are
$v_{i-1}$ and $v_i$.  Since $S$ is a stopping set, the check $f_i$
cannot have exactly one neighbor in $S$.  Therefore either both
$v_{i-1}$ and $v_i$ belong to $S$, or neither does. Equivalently,
\[
    v_{i-1}\in S
    \quad\Longleftrightarrow\quad
    v_i\in S.
\]
Removing the check $f_j$ from the alternating cycle leaves a path
through all of $v_1,\ldots,v_s$.  The equivalence above therefore
implies that either all of these variables belong to $S$ or none of
them does.

The latter possibility is impossible.  If
$v_1,\ldots,v_s\notin S$, then among the neighbors of $f_j$ the only
one in $S$ is $u\in A_j$, contradicting the stopping-set property.
Therefore $v_1,\ldots,v_s\in S$.
\end{proof}

The lemma will be applied with $S\subseteq R_0(E_\epsilon)$.
Consequently, if $K$ is a connected cyclic skeleton with
$\mu(K)=1$ and
\(
    |S\cap\supp(h_K)|=1,
\)
then every variable node of $B(K)$ must belong to
$R_0(E_\epsilon)$. We now bound the union
of the supports of all skeletons that satisfy this containment
condition.

For $R\subseteq V_n$ and $W\geq1$, let
$X_{\mathrm{cyc}}(W;R)$ be the union of $\supp(h_K)$ over all
connected cyclic skeletons $K$ such that
\(
    \mu(K)=1,
    \wt(h_K)\leq W,
\)
and every variable node of $B(K)$ belongs to $R$.  For the random
set left after peeling at erasure probability $\epsilon$, write
\[
    X_{\mathrm{cyc}}^{(\epsilon)}(W)
    :=
    X_{\mathrm{cyc}}\bigl(W;R_0(E_\epsilon)\bigr).
\]

Fix $\epsilon>\epsBP$, and let $(x,y)$ be the nonzero fixed point
reached by density evolution initialized at $x_0=\epsilon$.  Set
\(
    p:=\epsilon y^{\ell-2}=x/y
\)
and
\(
    c_{\mathrm{gate}}(\epsilon):=q/\log(bp).
\)
We have the following estimation on the size of \(X_{\mathrm{cyc}}^{(\epsilon)}(W)\).

\begin{lemma}
\label{lem:gated-active-zone}
Fix $\epsilon>\epsBP$ and suppose that
\(
    W_n=c\log n+O(1).
\)
Then, for all sufficiently large \(n\),
\begin{equation}
    \E
    |X_{\mathrm{cyc}}^{(\epsilon)}(W_n)|
    \leq
    n^{c/c_{\mathrm{gate}}(\epsilon)+o(1)}
    +
    n^{c/c_0-1+o(1)}.
    \label{eq:gated-cycle-asymptotic}
\end{equation}
Consequently, if
\(
    c<\min\{c_{\mathrm{gate}}(\epsilon),2c_0\},
\)
then
\[
    \frac{|X_{\mathrm{cyc}}^{(\epsilon)}(W_n)|}{n}
    \xrightarrow{\Prob}0.
\]
\end{lemma}

The proof is given in Appendix~\ref{app:dual-checks}.

It remains to control connected cyclic skeletons with
\(\mu(K)\geq2\).  For \(W\geq1\), let \(X_{\geq2}(W)\) be the union
of \(\supp(h_K)\) over all such skeletons satisfying
\(|K|\leq W\) and \(\wt(h_K)\leq W\).

\begin{lemma}
\label{lem:higher-excess-support}
Suppose that
\(
    W_n=c\log n+O(1)
\)
for a fixed \(c>0\).  Then, for all sufficiently large \(n\),
\[
    \E |X_{\geq2}(W_n)|
    \leq
    n^{c/c_0-1+o(1)}.
\]
Consequently, if \(c<2c_0\), then
\[
    \frac{|X_{\geq2}(W_n)|}{n}
    \xrightarrow{\Prob}0.
\]
\end{lemma}

The proof is also given in Appendix~\ref{app:dual-checks}.

We now combine the two types of cyclic skeleton.  For
\(R\subseteq V_n\), define
\begin{equation}
    X(W;R)
    :=
    X_{\mathrm{cyc}}(W;R)
    \cup
    X_{\geq2}(W).
    \label{eq:cyclic-support-set}
\end{equation}
By Lemmas~\ref{lem:gated-active-zone}
and~\ref{lem:higher-excess-support}, if
\(
    W_n=c\log n+O(1)
\)
and
\(
    c<\min\{c_{\mathrm{gate}}(\epsilon),2c_0\},
\)
then
\[
    \frac{|X(W_n;R_0(E_\epsilon))|}{n}
    \xrightarrow{\Prob}0.
\]

The role of $X(W;R)$ is captured by the following consequence.

\begin{lemma}
\label{lem:closure-stopping}
Fix $A<\infty$.  For every deterministic sequence of positive
integers satisfying $W_n\leq A\log n$ for all sufficiently large
$n$, with probability $1-o(1)$ over $G_n$, the following holds
for every $R\subseteq V_n$ and every stopping set
$S\subseteq R$ of $H_{0,n}$:
\[
    S\cap X(W_n;R)=\varnothing
    \quad\Longrightarrow\quad
    S\text{ is a stopping set of }H_{\leq W_n}(C_n).
\]
\end{lemma}

\begin{proof}
By Theorem~\ref{thm:localization}, with probability $1-o(1)$ over
$G_n$, every nonzero $h\in C_n^\perp$ with $\wt(h)\leq W_n$
admits a set $J\subseteq F_n$ such that
\(
    h=h_J
\)
and
\(
    |J|\leq W_n.
\)
Fix a graph for which this conclusion holds, and suppose that
$S\subseteq R$ is a stopping set of $H_{0,n}$ disjoint from
$X(W_n;R)$.

If $S$ were not a stopping set of $H_{\leq W_n}(C_n)$, some
nonzero $h\in C_n^\perp$ with $\wt(h)\leq W_n$ would satisfy
\(
    |S\cap\supp(h)|=1.
\)
Choose $J$ as above.  The reduction already established in this
subsection then gives a connected cyclic skeleton $K$ satisfying
\(
    |S\cap\supp(h_K)|=1,
    |K|\leq W_n,
\)
and
\(
    \wt(h_K)\leq W_n.
\)
If $\mu(K)=1$, Lemma~\ref{lem:cycle-gate} implies that every
variable node of $B(K)$ belongs to $S\subseteq R$, and hence
\(
    \supp(h_K)\subseteq X_{\mathrm{cyc}}(W_n;R).
\)
If $\mu(K)\geq2$, then
\(
    \supp(h_K)\subseteq X_{\geq2}(W_n)
\)
by definition.  In either case,
\(
    \supp(h_K)\subseteq X(W_n;R),
\)
contradicting
\(
    |S\cap\supp(h_K)|=1.
\)
\end{proof}

\subsection{Residual stability under deletions}

Fix $\epsilon>\epsBP$ and write
\(
    R:=R_0(E_\epsilon).
\)
Although $R$ is a stopping set of $H_{0,n}$, removing some of its
coordinates and running the peeling process again can remove
additional coordinates.  The next theorem shows that this effect is
controlled when the initial deletion is sufficiently small.

\begin{theorem}
\label{thm:core-resilience}
Fix $\epsilon\in(\epsBP,1]$ and let
\(
    R:=R_0(E_\epsilon).
\)
There are constants
$\delta_{\mathrm{seed}},\eta_{\mathrm{core}}>0$, depending only on
$(\ell,r,\epsilon)$, such that, with probability $1-o(1)$,
simultaneously for every $D\subseteq R$,
\begin{equation}
    |D|\leq\delta_{\mathrm{seed}}n
    \quad\Longrightarrow\quad
    |\Core_{H_{0,n}}(R\setminus D)|
    \geq\eta_{\mathrm{core}}n.
    \label{eq:linear-resilience-main}
\end{equation}

\end{theorem}

The proof is given in Appendix~\ref{app:residual-stability}.
The statement is simultaneous over all eligible deletion sets \(D\);
in particular, \(D\) may depend on the realized Tanner graph and
erasure pattern.

\subsection{Large stopping sets and the main proofs}

We now combine the two ingredients established in
Sections~\ref{sec:proof-locality-barrier}-A and~\ref{sec:proof-locality-barrier}-B.  We first prepare the large stopping set
used in the proofs of Theorems~\ref{thm:threshold}
and~\ref{thm:coverable}.  Only the case
$0<\gamma<c_*$ requires consideration: if $\gamma\geq c_*$, there
is no sequence satisfying
\(
    r\leq W_n\leq(c_*-\gamma)\log n
\)
for all sufficiently large $n$.
Assume therefore that $0<\gamma<c_*$, and let $(W_n)$ satisfy
\(
    r\leq W_n\leq(c_*-\gamma)\log n
\)
eventually.  Set
\(
    \bar c:=c_*-\gamma/2
\)
and
\(
    \widehat W_n:=\lfloor\bar c\log n\rfloor.
\)
Then $W_n\leq\widehat W_n$ eventually and
\(
    \bar c
    <
    \min\{c_{\mathrm{gate},\BP},2c_0\}.
\)
Let \((x(\epsilon),y(\epsilon))\) denote the upper nonzero
density-evolution fixed point.  By the continuity of the upper
fixed-point branch recalled in Section~\ref{sec:model},
\(
    x(\epsilon)\to x_{\BP}
\)
and
\(
    y(\epsilon)\to y_{\BP}
\)
as
\(
    \epsilon\downarrow\epsBP.
\)
At the BP tangency,
\(
    1
    =
    f_{\epsBP}'(x_{\BP})
    =
    b\frac{x_{\BP}}{y_{\BP}}(1-x_{\BP})^q.
\)
Therefore
\[
    c_{\mathrm{gate}}(\epsilon)
    =
    \frac{q}{\log\!\bigl(bx(\epsilon)/y(\epsilon)\bigr)}
    \longrightarrow
    \frac{q}{\log(bx_{\BP}/y_{\BP})}
    =
    \frac{1}{-\log(1-x_{\BP})}
    =
    c_{\mathrm{gate},\BP}.
\]

Now fix any $\epsilon_0>\epsBP$ such that
\(
    \bar c<c_{\mathrm{gate}}(\epsilon_0).
\)
Let $\delta_{\mathrm{seed}}$ and $\eta_{\mathrm{core}}$ be the
constants from Theorem~\ref{thm:core-resilience} at
$\epsilon_0$, and define
\[
    R:=R_0(E_{\epsilon_0}),
    \qquad
    D:=X(W_n;R)\cap R,
    \qquad
    S:=\Core_{H_{0,n}}(R\setminus D).
\]
Since $W_n\leq\widehat W_n$ and $X(W;R)$ is monotone in $W$, the
estimates of Section~\ref{sec:proof-locality-barrier}-A applied at
\(
    \widehat W_n=\bar c\log n+O(1)
\)
give
\(
    |X(W_n;R)|/n\xrightarrow{\Prob}0.
\)
Hence $|D|\leq\delta_{\mathrm{seed}}n$ with probability tending to
one, and Theorem~\ref{thm:core-resilience} gives
\(
    |S|\geq\eta_{\mathrm{core}}n
\)
with probability tending to one.  By construction,
\(
    S\subseteq R
\)
and
\(
    S\cap X(W_n;R)=\varnothing.
\)
Lemma~\ref{lem:closure-stopping} therefore shows that $S$ is also a
stopping set of $H_{\leq W_n}(C_n)$.
Thus, for every fixed $\epsilon_0>\epsBP$ satisfying
\(
    \bar c<c_{\mathrm{gate}}(\epsilon_0),
\)
there is a constant $\eta(\epsilon_0)>0$ such that, with probability
tending to one, $H_{\leq W_n}(C_n)$ has a stopping set $S$ with
\(
    |S|\geq\eta(\epsilon_0)n
\)
and
\(
    S\subseteq R_0(E_{\epsilon_0}).
\)

We now prove the main results stated in
Section~\ref{sec:barrier}.

\begin{proof}[Proof of Theorem~\ref{thm:threshold}]
Fix first $\epsilon<\epsBP$.  Since $W_n\geq r$ eventually,
$H_{\leq W_n}(C_n)$ contains every row of $H_{0,n}$.  Therefore,
for every erasure set $E$,
\(
    \Core_{H_{\leq W_n}(C_n)}(E)
    \subseteq
    \Core_{H_{0,n}}(E).
\)
Recall from Section~\ref{sec:model} that, for every fixed
$\epsilon<\epsBP$,
\(
    |R_0(E_\epsilon)|/n\xrightarrow{\Prob}0.
\)
The same is therefore true for
\(
    \Core_{H_{\leq W_n}(C_n)}(E_\epsilon),
\)
i.e.,
\[
    \frac{|\Core_{H_{\leq W_n}(C_n)}(E_\epsilon)|}{n}
    \xrightarrow{\Prob}0.
\]

Now fix $\epsilon>\epsBP$.  Choose
\(
    \epsilon_1\in(\epsBP,\epsilon)
\)
sufficiently close to $\epsBP$ that
\(
    \bar c<c_{\mathrm{gate}}(\epsilon_1).
\)
Applying the large-stopping-set argument above with
$\epsilon_0=\epsilon_1$, there is a constant $\eta_1>0$ such that,
with probability tending to one, $H_{\leq W_n}(C_n)$ has a stopping
set $S$ satisfying
\(
    |S|\geq\eta_1n
\)
and
\(
    S\subseteq R_0(E_{\epsilon_1})
    \subseteq E_{\epsilon_1}.
\)

Couple the two erasure sets monotonically.  Let
$(U_v)_{v\in V_n}$ be i.i.d.\ uniform random variables on $[0,1]$,
independent of $G_n$, and define
\(
    E_\alpha:=\{v\in V_n:U_v\leq\alpha\},
    \alpha\in[0,1].
\)
For each $\alpha$, the set $E_\alpha$ has the BEC$(\alpha)$
distribution, and since $\epsilon_1<\epsilon$,
\(
    E_{\epsilon_1}\subseteq E_\epsilon
\)
almost surely.  Therefore, whenever the stopping set $S$ above
exists,
\(
    S\subseteq E_{\epsilon_1}\subseteq E_\epsilon
\)
almost surely.
Since $S$ is a stopping set of $H_{\leq W_n}(C_n)$ and
$\Core_{H_{\leq W_n}(C_n)}(E_\epsilon)$ is the largest such stopping
set contained in $E_\epsilon$, we have
\(
    \left|
        \Core_{H_{\leq W_n}(C_n)}(E_\epsilon)
    \right|
    \geq |S|
    \geq\eta_1 n.
\)
Thus the normalized residual for $H_{\leq W_n}(C_n)$ does not
converge to zero at erasure probability $\epsilon$.

Thus every fixed $\epsilon<\epsBP$ belongs to the set defining
$\epsilon_{\mathrm{res}}(W_\bullet)$, whereas every fixed
$\epsilon>\epsBP$ lies outside it.  Hence
\[
    \epsilon_{\mathrm{res}}(W_\bullet)=\epsBP.
\]
\end{proof}

For the proof of Theorem~\ref{thm:coverable}, let
\(\epsMAP=\epsMAP(\ell,r)\) denote the bit-MAP threshold of the
\((\ell,r)\)-regular ensemble on the BEC.  For
\(3\leq\ell<r\), the regular ensemble has the standard strict
BP--MAP gap~\cite{KudekarRichardsonUrbanke2011,RichardsonUrbanke2008},
namely, \(\epsBP<\epsMAP\).

We also need the following standard consequence of bit-MAP decoding
below the MAP threshold and the linear minimum distance of the regular
ensemble.

\begin{lemma}
\label{lem:submap-unique}
For every fixed
\(\epsilon_0\in(\epsBP,\epsMAP)\),
\[
    \Prob\!\left\{
        E_{\epsilon_0}
        \text{ is not uniquely decodable for }C_n
    \right\}
    \longrightarrow0.
\]
\end{lemma}

\begin{proof}
For a code \(C\) and an erasure set \(E\), let
\(\mathsf{Amb}(C,E)\subseteq E\) denote the set of coordinates
that are not uniquely determined by the unerased coordinates.
Equivalently,
\(
    i\in\mathsf{Amb}(C,E)
\)
if and only if there is a nonzero \(z\in C\) with \(\supp(z)\subseteq E\) and \(z_i=1\).

For every fixed \(\epsilon_0<\epsMAP\), the finite-length BEC
EXIT identity and the bit-MAP threshold convention of~\cite[Sec.~II and Definition~2]{MeassonMontanariUrbanke2008}
give
\[
    \frac{1}{n}
    \E\bigl|\mathsf{Amb}(C_n,E_{\epsilon_0})\bigr|
    \longrightarrow0
\]
in the regular configuration ensemble.  Definition~2 in~\cite{MeassonMontanariUrbanke2008} uses a
\(\limsup\) in the blocklength, so its vanishing below
\(\epsMAP\) gives convergence of this nonnegative normalized
quantity to zero.  Conditioning on simplicity then gives the same
conclusion for the ensemble used here, as in
Appendix~\ref{app:dual-checks}.
Consequently, for every fixed \(\delta>0\), Markov's inequality gives
\[
    \Prob\!\left\{
        |\mathsf{Amb}(C_n,E_{\epsilon_0})|
        \geq\delta n
    \right\}
    \longrightarrow0.
\]
For the binary \((\ell,r)\)-regular configuration ensemble, the
individual-code minimum-distance estimate for regular LDPC ensembles
applies with \(q=2\), \(c=\ell\), and \(d=r\)~\cite[Theorem~6.2, Remark~6.3, and (51)]{YangEtAl2011}.
In particular, for every fixed
\(
    0<\delta_{\mathrm{dist}}<x_0(2,\ell,r),
\)
\[
    \Prob\!\left\{
        d(C_n)\leq\delta_{\mathrm{dist}}n
    \right\}
    \longrightarrow0.
\]
For \(q=2\), the ensemble in~\cite{YangEtAl2011} is the regular
bipartite configuration model used here before conditioning on
simplicity: the edge labels are trivial, while parallel edges are
allowed.  Conditioning on simplicity therefore gives the same
conclusion for the uniform simple ensemble, by the simplicity
transfer in Appendix~\ref{app:dual-checks}.

Now suppose that \(E_{\epsilon_0}\) is not uniquely decodable.
Then there is a nonzero codeword \(z\in C_n\) with
\(\supp(z)\subseteq E_{\epsilon_0}\).  If
\(d(C_n)\geq\delta_{\mathrm{dist}}n\), then
\(
    |\supp(z)|\geq\delta_{\mathrm{dist}}n.
\)
Every coordinate in \(\supp(z)\) is ambiguous, and hence
\(
    |\mathsf{Amb}(C_n,E_{\epsilon_0})|
    \geq\delta_{\mathrm{dist}}n.
\)
Therefore
\[
    \{E_{\epsilon_0}\text{ is not uniquely decodable}\}
    \subseteq
    \{d(C_n)<\delta_{\mathrm{dist}}n\}
    \cup
    \{|\mathsf{Amb}(C_n,E_{\epsilon_0})|
        \geq\delta_{\mathrm{dist}}n\}.
\]
Both events on the right have probability tending to zero, which
proves the claim.
\end{proof}

\begin{proof}[Proof of Theorem~\ref{thm:coverable}]
Choose
\(
    \epsilon_0\in(\epsBP,\epsMAP)
\)
sufficiently close to $\epsBP$ that
\(
    \bar c<c_{\mathrm{gate}}(\epsilon_0).
\)
Applying the same large-stopping-set argument at this value of
$\epsilon_0$, there is a constant $\eta_0>0$ such that, with
probability tending to one, $H_{\leq W_n}(C_n)$ has a stopping set
$S$ satisfying
\(
    |S|\geq\eta_0n
\)
and
\(
    S\subseteq R_0(E_{\epsilon_0})
    \subseteq E_{\epsilon_0}.
\)
By Lemma~\ref{lem:submap-unique}, with probability tending to one,
\(E_{\epsilon_0}\) is uniquely decodable for \(C_n\).  Hence
\(
    \rank\bigl((H_{0,n})_{E_{\epsilon_0}}\bigr)
    =
    |E_{\epsilon_0}|.
\)
Since
\(
    S\subseteq E_{\epsilon_0},
\)
the columns of $(H_{0,n})_S$ are a subfamily of a linearly
independent family.  Therefore
\(
    \rank((H_{0,n})_S)=|S|.
\)
For all sufficiently large $n$, we have $W_n\geq r$.  Applying
Proposition~\ref{prop:budget-domination} with
$H=H_{0,n}$ and $W=W_n$ gives
\[
    \rank\bigl((H_{\leq W_n}(C_n))_S\bigr)
    =
    \rank((H_{0,n})_S)
    =
    |S|.
\]
Thus $S$ is a coverable stopping set of
$H_{\leq W_n}(C_n)$.

It remains to remove the auxiliary erasure randomness.  The
existence of a coverable stopping set of size at least $\eta_0 n$
is a property of $G_n$ alone.  Since the joint choice of
$(G_n,E_{\epsilon_0})$ produces such a set with probability tending
to one, the probability over $G_n$ that no such set exists tends to
zero.  Taking $\eta=\eta_0$ proves the theorem.
\end{proof}

Theorem~\ref{thm:heavy-support} uses the same two ingredients in a
slightly different way.  Suppose that the rows heavier than $a_n$
collectively meet only a small set of coordinates.  We include
those coordinates in the deletion set before applying
Theorem~\ref{thm:core-resilience}.  The resulting stopping set
avoids every heavier row, while Lemma~\ref{lem:closure-stopping}
controls the remaining rows.

\begin{proof}[Proof of Theorem~\ref{thm:heavy-support}]
Recall that
\(
    a_n=(c_*-\gamma)\log n.
\)
If $\gamma\geq c_*$, then $a_n\leq0$.  For every parity-check
matrix $H$ for $C_n$, the union of the supports of its nonzero rows
is $V_n$: otherwise $\operatorname{rowspan}H=C_n^\perp$ would
vanish at some coordinate, whereas $H_{0,n}$ has a row containing
that coordinate.  Hence $U_{>a_n}(H)=V_n$, and the theorem is
immediate.  We therefore assume $0<\gamma<c_*$.

Set
\(
    W_n^-:=\lfloor a_n\rfloor,
    \bar c:=c_*-\gamma/2,
\)
and
\(
    \widehat W_n:=\lfloor\bar c\log n\rfloor.
\)
Since row weights are integers, the rows of weight at most $a_n$
are precisely those of weight at most $W_n^-$.  Moreover,
\(
    W_n^-\leq\widehat W_n
\)
for all sufficiently large $n$.
Choose
\(
    \epsilon_0\in(\epsBP,\epsMAP)
\)
sufficiently close to $\epsBP$ that
\(
    \bar c<c_{\mathrm{gate}}(\epsilon_0).
\)
Let
\(
    R:=R_0(E_{\epsilon_0}).
\)
By monotonicity in the weight budget and the estimates of
Section~\ref{sec:proof-locality-barrier}-A applied at $\widehat W_n$,
\(
    |X(W_n^-;R)|/n
    \xrightarrow{\Prob}0.
\)

Let $\delta_{\mathrm{seed}}>0$ be the constant from
Theorem~\ref{thm:core-resilience} at $\epsilon_0$, and set
\(
    \delta:=\delta_{\mathrm{seed}}/2.
\)
Since
\(
    |X(W_n^-;R)|/n\xrightarrow{\Prob}0,
\)
\(
    |X(W_n^-;R)|\leq\delta n
\)
with probability tending to one.
Intersect this event with the high-probability events from
Lemma~\ref{lem:closure-stopping} at the budget \(W_n^-\),
Theorem~\ref{thm:core-resilience} at \(\epsilon_0\), and
Lemma~\ref{lem:submap-unique}.  Fix a realization in this
intersection.

Suppose, for contradiction, that there exists an ML-exact
parity-check matrix $H$ for $C_n$ such that
\(
    |U_{>a_n}(H)|<\delta n.
\)
Write
\(
    U_H:=U_{>a_n}(H)
\)
and define
\(
    D_H
    :=
    \bigl(X(W_n^-;R)\cup U_H\bigr)\cap R
\)
and
\(
    S_H
    :=
    \Core_{H_{0,n}}(R\setminus D_H).
\)
Since
\(
    |D_H|
    \leq
    |X(W_n^-;R)|+|U_H|
    <
    2\delta n
    =
    \delta_{\mathrm{seed}}n,
\)
Theorem~\ref{thm:core-resilience} gives
\(
    |S_H|\geq\eta_{\mathrm{core}}n.
\)
By construction,
\(
    S_H\cap X(W_n^-;R)=\varnothing
\)
and
\(
    S_H\cap U_H=\varnothing.
\)
Since $S_H$ is a stopping set of $H_{0,n}$ contained in $R$,
Lemma~\ref{lem:closure-stopping}, applied at the budget $W_n^-$,
shows that $S_H$ is also a stopping set of
$H_{\leq W_n^-}(C_n)$.

Now consider any nonzero row $h$ of $H$.  If
\(
    \wt(h)\leq a_n,
\)
then, since row weights are integers,
\(
    \wt(h)\leq W_n^-.
\)
Hence $h$ is a row of $H_{\leq W_n^-}(C_n)$ and cannot meet $S_H$
in exactly one coordinate.  If instead
\(
    \wt(h)>a_n,
\)
then
\(
    \supp(h)\subseteq U_H,
\)
so $h$ does not meet $S_H$ at all.  Thus $S_H$ is also a stopping set
of $H$.
Moreover,
\(
    S_H\subseteq R \subseteq E_{\epsilon_0}.
\)
Since $E_{\epsilon_0}$ is uniquely decodable for $C_n$, the columns
of $(H_{0,n})_{S_H}$ are linearly independent, and hence
\(
    \rank((H_{0,n})_{S_H})=|S_H|.
\)
Because
\(
    \operatorname{rowspan}H
    =
    C_n^\perp
    =
    \operatorname{rowspan}H_{0,n},
\)
restriction to the coordinates in $S_H$ gives
\(
    \rank(H_{S_H})=|S_H|.
\)
Thus $S_H$ is a nonempty stopping set of $H$ with
\(
    \rank(H_{S_H})=|S_H|.
\)
By Proposition~\ref{prop:rank-interface}, $H$ is not ML-exact on
the BEC.  This contradicts the assumption that $H$ is ML-exact.
Therefore, on the fixed high-probability event, every ML-exact
parity-check matrix $H$ for $C_n$ satisfies
\[
    |U_{>a_n}(H)|\geq\delta n.
\]

Finally,
the high-probability event used above depends on
$(G_n,E_{\epsilon_0})$, whereas the conclusion of the theorem is
stated only in terms of $G_n$.
If a graph admitted an ML-exact parity-check matrix $H$ with
\(
    |U_{>a_n}(H)|<\delta n,
\)
then none of the auxiliary erasure patterns for that graph could
belong to the high-probability event used above, because the
preceding deterministic argument would give a contradiction.
Therefore the probability of such a graph is bounded by the failure
probability of that joint event, which tends to zero.  This proves
the theorem.
\end{proof}

\section{Concluding remarks}
\label{sec:discussion}
\label{sec:conclusion}

We have shown that, for random regular LDPC codes on the BEC, allowing
arbitrarily many redundant checks does not remove the gap between peeling
and unique recovery when their row weights are restricted to the
logarithmic range considered here.  Even the complete family of all dual
words below the weight cutoff leaves the normalized-residual threshold at
\(\epsilon_{\mathrm{BP}}\) and still has linear-size stopping sets that are
uniquely decodable.  Thus row weight remains a genuine constraint after
the number and selection of redundant checks have been removed from the
problem.

The representation results make this constraint more global.  Any
parity-check representation whose peeling decoder succeeds on every
uniquely decodable erasure pattern must contain checks above the
logarithmic scale, and the supports of those checks must collectively
involve a linear number of coordinates.  At finite lengths, low-weight
redundant checks can nevertheless provide substantial gains, as illustrated
in Fig.~\ref{fig:correctable-erasure-failure}.  The finite-length benefit
and the asymptotic exactness requirement therefore describe different
aspects of the same row-weight restriction.

The asymptotic order of the row weight required for exact peeling remains
open.  Corollary~\ref{cor:wex} gives a logarithmic lower bound with high
probability, whereas the general weight-\((k+1)\) construction gives a
linear upper bound at fixed rate.  The coefficient \(c_*\) specifies the
range established by the present analysis; determining the asymptotic
order of \(\wex(C_n)\) is a natural next question.

\appendix
\subsection{Localization and cyclic structure of low-weight dual checks}
\label{app:dual-checks}

This appendix supplies the proofs deferred from
Section~\ref{sec:proof-locality-barrier}-A.  We prove the localization
of low-weight dual checks, the reduction of an exact-one intersection
to a connected cyclic skeleton, the structural identities for such
skeletons, and the separate probabilistic bounds for the cases
\(\mu=1\) and \(\mu\geq2\).

The probabilistic estimates below are first carried out in the
bipartite configuration model.  Each variable node has \(\ell\)
labeled half-edges and each check node has \(r\) labeled half-edges,
with \(\ell n=rm\).  A uniformly random perfect matching between the
two sides produces a bipartite multigraph.  We interpret incidences
with multiplicity and take all row sums over \(\F_2\).  On the event
that the multigraph is simple, this agrees with the Tanner-graph and
binary-matrix conventions used in
Sections~\ref{sec:model}--\ref{sec:proof-locality-barrier}.
For \(x\in\mathbb N\) and \(k\geq0\), we write
\((x)_k:=x(x-1)\cdots(x-k+1)\), with \((x)_0:=1\), for the
falling factorial.

Conditioning the configuration model on simplicity gives the uniform
simple ensemble of Section~\ref{sec:model}.  For fixed \((\ell,r)\),
both degree sequences are uniformly bounded, so the conditions for a
simplicity probability bounded away from zero are satisfied~\cite[Theorem~1.1 and Remark~1.3]{BlanchetStauffer2013}.
Thus an event of probability \(o(1)\) in the configuration model still
has probability \(o(1)\) after conditioning on simplicity.

We now give the proof of Theorem~\ref{thm:localization}.

\begin{proof}[Proof of Theorem~\ref{thm:localization}]
Write \(W=W_n\).  If the conclusion fails, then some nonzero
\(h\in C_n^\perp\) with \(\wt(h)\leq W\) has a minimum generating
set \(J\subseteq F_n\) of size \(s=\sigma(h)>W\).  Since
\(J\subseteq F_n\), we always have \(s\leq m=\ell n/r\), so there
are only \(O(n)\) possible values of \(s\).  We rule out such
minimum generating sets in two ranges: first
\(W<s\leq\delta_{\mathrm{loc}}n\), and then
\(s>\delta_{\mathrm{loc}}n\).

For \(s>W\), let
\(\mathcal J_{s,W}:=\{J\subseteq F_n:|J|=s,\ \wt(h_J)\leq W\}\)
and \(G_{s,W}:=|\mathcal J_{s,W}|\).
Fix \(J\in\mathcal J_{s,W}\), let \(t\) be the number of variable
nodes incident to the checks in \(J\), and write \(w=\wt(h_J)\).
Among these \(t\) variables, exactly \(w\) have odd degree in \(T(J)\).
The remaining \(t-w\) variables have positive even degree in \(T(J)\),
and hence degree in \(T(J)\) at least two.  Since the checks in \(J\)
have \(rs\) incident half-edges,
\(rs\geq w+2(t-w)=2t-w\).  In particular,
\(t\leq(rs+W)/2\).
For a fixed set \(J\), its \(rs\) check-side half-edges are matched
uniformly to an \(rs\)-subset of the \(\ell n\) variable-side
half-edges.  Therefore
\begin{equation}
 \E G_{s,W}
 \leq
 \binom ms
 \sum_{t\leq(rs+W)/2}
 \binom nt
 \frac{\binom{\ell t}{rs}}{\binom{\ell n}{rs}} .
 \label{eq:app-local-occupancy}
\end{equation}
A nonzero summand necessarily satisfies \(\ell t\geq rs\), because
the \(t\) touched variables contain only \(\ell t\) half-edges.
Thus \(t\geq rs/\ell\).  Since \(W<s\), we also have
\(t<(r+1)s/2\), so \(u:=t/s\) is bounded above and below by positive
constants depending only on \((\ell,r)\).
Using
\(\binom ms\leq(em/s)^s\),
\(\binom nt\leq(en/t)^t\), and
\(
 \binom{\ell t}{rs}/\binom{\ell n}{rs}
 \leq \left(t/n\right)^{rs},
\)
each nonzero summand in \eqref{eq:app-local-occupancy} is at most
\(\bigl((e\ell/r)e^u u^{r-u}\bigr)^s
 (s/n)^{rs-s-t}\).
Since \(u\) ranges over a fixed interval depending only on
\((\ell,r)\), the factor
\((e\ell/r)e^u u^{r-u}\) is uniformly bounded by a constant
\(C_{\mathrm{loc}}=C_{\mathrm{loc}}(\ell,r)\).
There are \(O(s)\) possible values of \(t\), which can be absorbed
by increasing \(C_{\mathrm{loc}}\).  Since
\(t<(r+1)s/2\),
\(
 rs-s-t\geq (r-3)s/2,
\)
and hence
\[
 \E G_{s,W}
 \leq
 \left[
 C_{\mathrm{loc}}
 \left(\frac{s}{n}\right)^{(r-3)/2}
 \right]^s .
\]
Choose \(\delta_{\mathrm{loc}}>0\) sufficiently small that
\(r\delta_{\mathrm{loc}}/\ell<1/2\) and
\(C_{\mathrm{loc}}\delta_{\mathrm{loc}}^{(r-3)/2}<1/2\).
Since \(r\geq4\), the preceding estimate gives
\[
 \sum_{W<s\leq\delta_{\mathrm{loc}}n}\E G_{s,W}=o(1).
\]
Indeed, if \(s\leq\sqrt n\), then
\(C_{\mathrm{loc}}(s/n)^{(r-3)/2}
 \leq C_{\mathrm{loc}}n^{-(r-3)/4}=o(1)\).
If instead
\(\sqrt n<s\leq\delta_{\mathrm{loc}}n\), then
\(C_{\mathrm{loc}}(s/n)^{(r-3)/2}
 \leq C_{\mathrm{loc}}\delta_{\mathrm{loc}}^{(r-3)/2}<1/2\).
Summing the corresponding bounds on \(\E G_{s,W}\) over the two
ranges proves the estimate above.

It remains to consider generating sets with
\(s>\delta_{\mathrm{loc}}n\).  Write \(s=\alpha m\).  For a fixed
set of \(s\) checks, the \(rs= \alpha rm = \alpha\ell n\) half-edges incident to
those checks are matched to a uniformly random
\(\alpha\ell n\)-subset of the \(\ell n\) variable-side half-edges.
Let \(N_{\alpha,W}\) denote the number of such subsets
\(Q\) for which at most \(W\) variable nodes are incident to an odd
number of half-edges in \(Q\).  These are precisely the variable
nodes belonging to \(\supp(h_J)\), and hence
\[
    \E G_{s,W}
    =
    \binom ms
    \frac{N_{\alpha,W}}
         {\binom{\ell n}{\alpha\ell n}}.
\]
To bound \(N_{\alpha,W}\), label the \(\ell\) half-edges at every
variable node by \(1,\ldots,\ell\).  For a subset \(Q\) counted by
\(N_{\alpha,W}\), let \(S_i\subseteq V_n\) be the set of variables
whose \(i\)-th half-edge belongs to \(Q\).  Then
\(|S_1|+\cdots+|S_\ell|=\alpha\ell n\), and the variables incident to
an odd number of half-edges in \(Q\) form the symmetric difference
\(S_1\triangle\cdots\triangle S_\ell\).
Put \(k_i:=|S_i|\) and \(\beta_i:=k_i/n\), so that
\(\beta_1+\cdots+\beta_\ell=\alpha\ell\).  For fixed
\(k_1,\ldots,k_\ell\), choose \(j\) so that
\(h(\beta_j)\) is maximal, where
\(h(x)=-x\log x-(1-x)\log(1-x)\) is the binary entropy.
Once the symmetric difference
\(B:=S_1\triangle\cdots\triangle S_\ell\) and the sets
\(S_i\), \(i\ne j\), are specified, the remaining set is determined
by
\(S_j=B\triangle(\bigtriangleup_{i\ne j}S_i)\).
Thus the number of choices with these cardinalities and
\(|B|\leq W\) is at most
\[
    \left(\sum_{w=0}^{W}\binom nw\right)
    \prod_{i\ne j}\binom n{k_i}.
\]
Because \(W=O(\log n)\),
\(\sum_{w=0}^{W}\binom nw\leq(W+1)n^W=e^{o(n)}\).
Also, \(\binom n{k_i}\leq e^{nh(\beta_i)}\).  By the choice of
\(j\),
\[
    \sum_{i\ne j}h(\beta_i)
    \leq
    \frac{\ell-1}{\ell}\sum_{i=1}^{\ell}h(\beta_i)
    \leq
    (\ell-1)h(\alpha),
\]
where the second inequality follows from the concavity of \(h\) and
\(\sum_i\beta_i=\alpha \ell \).  There are at most
\((n+1)^\ell=e^{o(n)}\) possible choices of the cardinalities
\(k_1,\ldots,k_\ell\).  Hence
\(
    N_{\alpha,W}
    \leq (n+1)^\ell (W+1)n^W e^{n(\ell-1)h(\alpha)} =
    e^{n(\ell-1)h(\alpha)+o(n)}.
\)

We next restrict the possible values of \(\alpha\) for a minimum
generating set.  Since \(s>\delta_{\mathrm{loc}}n\),
\(\alpha>r\delta_{\mathrm{loc}}/\ell\).  For the upper bound, sum
all \(m\) original check rows over \(\F_2\).  If \(\ell\) is even,
this sum is zero, so \(F_n\setminus J\) generates the same dual word
as \(J\).  Minimality of \(J\) therefore gives
\(\alpha\leq1/2\).
If \(\ell\) is odd, the sum of all original rows is \(\1\), so the
complement \(F_n\setminus J\) generates \(\1+h_J\).  Since \(\wt(h_J)\leq W\),
\(
    n-W
    \leq
    \wt(\1+h_J)
    \leq
    r(m-|J|) = r(m-s),
\)
and hence
\(\alpha\leq1-1/\ell+W/(\ell n)\).  As \(W=O(\log n)\), this is at
most \(1-1/(2\ell)\) for all sufficiently large \(n\).
Thus the relevant values of \(\alpha\) lie in a fixed closed
subinterval of \((0,1)\).
Stirling's formula can therefore be
applied uniformly in this range.  Using \(m=\ell n/r\) together
with the preceding bound on \(N_{\alpha,W}\), we obtain
\[
    \frac1n\log\E G_{s,W}
    \leq
    -\left(1-\frac{\ell}{r}\right)h(\alpha)+o(1).
\]
Since \(r>\ell\) and \(h(\alpha)\) is bounded away from zero on the
relevant interval, there is a constant \(\kappa>0\) such that
\(\E G_{s,W}\leq e^{-\kappa n}\) for all sufficiently large \(n\)
and every remaining value of \(s\).
Since there are only \(O(n)\) possible values of \(s\), summing
\(\E G_{s,W}\) over the values that can occur for a minimum
generating set in this range gives \(o(1)\).

Combining the two size ranges, the sum of \(\E G_{s,W}\) over all
sizes \(s>W\) that can occur for a minimum generator of a
counterexample is \(o(1)\).  By Markov's inequality, with probability
\(1-o(1)\) in the configuration model no such generating set exists.
Thus every nonzero dual word \(h\) with \(\wt(h)\leq W\) satisfies
\(\sigma(h)\leq W\).  Conditioning on simplicity gives the same
conclusion for the uniform simple ensemble.
\end{proof}

We now prove the part of the cyclic reduction used in
Section~\ref{sec:proof-locality-barrier}-A that is deferred to the
appendix.  The main text has already reduced an exact-one
intersection to a connected graph \(T(J)\).  Lemma~\ref{lem:tree-inert}
rules out the case in which \(T(J)\) is a tree.  Proposition~\ref{prop:stopping-active-reduction}
then removes the tree branches attached to the cyclic part and leaves
a connected cyclic skeleton.

\begin{lemma}
\label{lem:tree-inert}
If \(S\) is a stopping set of the original parity-check matrix and
\(T(J)\) is a tree, then
\(
    |S\cap\supp(h_J)|\ne1.
\)
\end{lemma}

\begin{proof}
Remove from \(T(J)\) all variable nodes that do not belong to \(S\),
together with their incident edges, and then remove the resulting
isolated check nodes.  Denote the remaining graph by \(T'\).
Since \(S\) is a stopping set, every check node has either no neighbor
in \(S\) or at least two neighbors in \(S\).  Therefore every check
remaining in \(T'\) has degree at least two.  Since \(T(J)\) is a
tree, \(T'\) is a forest.

If \(T'\) has no edge, then no variable in \(S\) is incident to a
check in \(J\), and hence
\(S\cap\supp(h_J)=\varnothing\).
Suppose instead that \(T'\) contains an edge.  Then some nontrivial
component of \(T'\) is a tree and therefore has at least two leaves.
No check node can be a leaf in \(T'\), since every remaining check
has degree at least two.  Hence these leaves are variable nodes.
For every variable \(v\in S\), all edges of \(T(J)\) incident to
\(v\) remain in \(T'\).  Thus every variable leaf of \(T'\) has
degree one in \(T(J)\).  Such a variable is incident to exactly one
check in \(J\), and therefore belongs to \(\supp(h_J)\).  Hence
\(
    |S\cap\supp(h_J)|\geq2.
\)
In either case,
\(
    |S\cap\supp(h_J)|\ne1.
\)
\end{proof}

\begin{proposition}
\label{prop:stopping-active-reduction}
Let \(S\) be a stopping set of the original parity-check matrix and
let \(J\subseteq F_n\) be such that \(T(J)\) is connected and
\(
    |S\cap\supp(h_J)|=1.
\)
Then there is a connected cyclic skeleton \(K\subseteq J\) such that
\[
    |S\cap\supp(h_K)|=1,
    \qquad
    |K|\leq |J|,
    \qquad
    \wt(h_K)\leq\wt(h_J).
\]
\end{proposition}

\begin{proof}
By Lemma~\ref{lem:tree-inert}, the connected graph \(T(J)\) cannot
be a tree, and hence contains a cycle.
Let \(\mathcal C\) be the two-core of \(T(J)\), obtained by repeatedly
deleting vertices of degree zero or one.  Since \(T(J)\) is connected
and contains a cycle, \(\mathcal C\) is nonempty and connected.
Every component of \(T(J)-V(\mathcal C)\) is a tree joined to
\(\mathcal C\) by exactly one edge.

We now show that every such tree component containing at least one
check can be removed.
We remove these components one at a time.  Initially set
\(\widehat J=J\).  Consider a component of
\(T(J)-V(\mathcal C)\) that contains at least one check, and let
\(L\) be the set of its check nodes.  Put
\(
    K':=\widehat J\setminus L.
\)
Then
\(
    \widehat J=K'\mathbin{\dot\cup}L.
\)
The component is joined to \(\mathcal C\) by a unique edge.
If the endpoint of this edge in \(\mathcal C\) is a variable node,
let \(v\) be that variable.  If that endpoint is a check node, let
\(v\) be the variable endpoint outside \(\mathcal C\).  In either
case \(T(L)\) is a tree and
\(
    N(K')\cap N(L)=\{v\}.
\)
Let \(t:=|L|\), and for \(1\leq a\leq\ell\), let \(n_a\) be the
number of variable nodes having degree \(a\) in \(T(L)\).  The
attachment variable \(v\) is included in this count, with its degree
measured in \(T(L)\). Every check in \(L\) has degree \(r\) in \(T(L)\),
so the handshaking lemma gives
\(
    rt=\sum_{a=1}^{\ell} a n_a.
\)
Since \(T(L)\) is a tree, Euler's identity also gives
\(
    rt
    =
    t+\sum_{a=1}^{\ell}n_a-1.
\)
Subtracting these two identities yields
\[
    \sum_{a=2}^{\ell}(a-1)n_a=t-1.
\]
A variable belongs to \(\supp(h_L)\) exactly when its degree in
\(T(L)\) is odd.  Hence
\[
    \wt(h_L)
    =
    \sum_{\substack{1\leq a\leq\ell\\ a\text{ odd}}} n_a
    =
    n_1+
    \sum_{a=2}^{\ell}\1\{a\text{ odd}\}n_a.
\]
Again from Euler's identity,
\(
    n_1
    =
    (r-1)t+1-\sum_{a=2}^{\ell}n_a.
\)
Using
\(
    t-1=\sum_{a=2}^{\ell}(a-1)n_a
\),
we therefore obtain
\[
    \wt(h_L)
    =
    (r-2)t+2
    +
    \sum_{a=2}^{\ell}
    \bigl(a-2+\1\{a\text{ odd}\}\bigr)n_a
    \geq
    (r-2)t+2
    =
    qt+2.
\]

It remains to check that removing \(L\) preserves the exact-one
intersection.  Put
\(
    A:=\supp(h_{K'})
\)
and
\(
    D:=\supp(h_L).
\)
Since \(N(K')\cap N(L)=\{v\}\), we have
\(
    A\cap D\subseteq\{v\}
\)
and
\(
    \supp(h_{\widehat J})=A\triangle D.
\)
Write
\[
    a:=|S\cap A|,
    \qquad
    d:=|S\cap D|,
    \qquad
    c:=|S\cap A\cap D|\in\{0,1\}.
\]
Since the current row meets \(S\) in exactly one coordinate,
\(
    a+d-2c=1.
\)
Because \(T(L)\) is a tree, Lemma~\ref{lem:tree-inert} gives
\(d\ne1\).  If \(c=0\), then \(a+d=1\), and therefore
\((a,d)=(1,0)\).  If \(c=1\), then \(a,d\geq1\) and
\(a+d=3\); since \(d\ne1\), necessarily
\((a,d)=(1,2)\).  Thus in both cases
\(
    |S\cap\supp(h_{K'})|=1.
\)
Hence removing the branch \(L\) preserves the exact-one
intersection.
The removal also decreases the row weight.  Indeed,
\(
    \wt(h_{\widehat J})
    =
    |A\triangle D|
    =
    |A|+|D|-2|A\cap D|
    \geq
    |A|+|D|-2
    \geq
    |A|+qt.
\)
Therefore
\(
    \wt(h_{K'})
    \leq
    \wt(h_{\widehat J})-qt.
\)
After removing \(L\), replace \(\widehat J\) by \(K'\) and proceed
to the next check-containing component.

Remove the check-containing components of
\(T(J)-V(\mathcal C)\) one at a time in this way.  After each
removal the exact-one intersection is preserved, while the number
of checks strictly decreases, so the procedure terminates.  Let
\(K\) be the remaining set of checks; equivalently, \(K\) is
the check-node set of \(\mathcal C\).

Every variable node of \(\mathcal C\) has degree at least two in
\(\mathcal C\).  Conversely, every variable of \(T(K)\) outside
\(\mathcal C\) has degree one in \(T(K)\).  Indeed, if such a
variable were adjacent to at least two checks in \(K\), adjoining it
and those edges to \(\mathcal C\) would give a larger subgraph of
minimum degree at least two, contradicting the maximality of the
two-core.  It follows from the definition of \(B(K)\) that
\(
    B(K)=\mathcal C.
\)
Thus \(B(K)\) is nonempty and connected and has minimum degree at
least two on both sides.  Hence \(K\) is a connected cyclic
skeleton.

Summing the weight decrease over the removed branches gives, in
particular,
\[
    |S\cap\supp(h_K)|=1,
    \qquad
    |K|\leq|J|,
    \qquad
    \wt(h_K)\leq\wt(h_J).
\]
\end{proof}

We next prove the structural identities used in
Section~\ref{sec:proof-locality-barrier}-A after the reduction to a
connected cyclic skeleton.  They relate the number of checks in the
connected cyclic skeleton, the weight of their binary sum, and the cyclomatic number
of the shared-incidence graph.  The first identity is the cyclic
analogue of the tree-branch calculation above, with the cyclomatic
number accounting for the additional edges.  These identities give
the size bound used in the case \(\mu\geq2\) and show that
\(\mu=1\) corresponds exactly to an alternating cycle.

\begin{theorem}
\label{thm:skeleton-identity}
Let \(K\) be a connected cyclic skeleton, and put
\(
    s:=|K|,
    w:=\wt(h_K),
    \mu:=|E(B(K))|-|V(B(K))|+1.
\)
For \(1\leq a\leq\ell\), let \(n_a\) be the number of variable nodes
having degree \(a\) in \(T(K)\).  Then
\begin{equation}
    w
    =
    qs+2-2\mu
    +
    \sum_{a=2}^{\ell}
    \bigl(a-2+\1\{a\text{ odd}\}\bigr)n_a.
    \label{eq:app-skeleton-weight}
\end{equation}
\begin{equation}
    \sum_{f\in K}\bigl(d_{B(K)}(f)-2\bigr)
    +
    \sum_{a=2}^{\ell}(a-2)n_a
    =
    2\mu-2.
    \label{eq:app-skeleton-excess}
\end{equation}
In particular,
\[
    1\leq\mu\leq(r-1)s,
    \qquad
    s\leq\frac{w+2\mu-2}{q}.
\]
If \(\mu=1\), then \(B(K)\) is an alternating cycle and
\(w=qs\).
\end{theorem}

\begin{proof}
Since \(B(K)\) contains precisely the variable nodes having degree
at least two in \(T(K)\),
\(
    |E(B(K))|
    =
    \sum_{a=2}^{\ell}a n_a,
    |V(B(K))|
    =
    s+\sum_{a=2}^{\ell}n_a.
\)
Hence the definition of the cyclomatic number gives
\[
    \mu
    =
    \sum_{a=2}^{\ell}(a-1)n_a-s+1.
\]
On the other hand, the \(s\) checks in \(K\) have \(rs\) incident
edges in \(T(K)\), and therefore
\(
    rs
    =
    \sum_{a=1}^{\ell}a n_a
    =
    n_1+\sum_{a=2}^{\ell}a n_a.
\)
A variable belongs to \(\supp(h_K)\) exactly when its degree in
\(T(K)\) is odd.  Thus
\[
    w
    =
    \sum_{\substack{1\leq a\leq\ell\\ a\text{ odd}}}n_a
    =
    n_1+
    \sum_{a=2}^{\ell}\1\{a\text{ odd}\}n_a.
\]
Eliminating \(n_1\) with the preceding edge count and then using the
formula for \(\mu\) gives
\[
    w
    =
    qs+2-2\mu
    +
    \sum_{a=2}^{\ell}
    \bigl(a-2+\1\{a\text{ odd}\}\bigr)n_a,
\]
which proves \eqref{eq:app-skeleton-weight}.

For the second identity, let \(V_K\) denote the variable-node set
of \(B(K)\).  Thus
\(
    V(B(K))=K\mathbin{\dot\cup}V_K.
\)
For \(v\in V_K\), its degree in \(B(K)\) is the same as its degree
in \(T(K)\).  Hence
\[
    \sum_{v\in V_K}\bigl(d_{B(K)}(v)-2\bigr)
    =
    \sum_{a=2}^{\ell}(a-2)n_a.
\]
Adding the corresponding contribution from the check nodes gives
\[
\begin{aligned}
    \sum_{f\in K}\bigl(d_{B(K)}(f)-2\bigr)
    +\sum_{a=2}^{\ell}(a-2)n_a &=
    \sum_{x\in V(B(K))}\bigl(d_{B(K)}(x)-2\bigr)\\
    &=
    2|E(B(K))|-2|V(B(K))|
    =
    2\mu-2,
\end{aligned}
\]
where the second equality is the handshaking lemma and the last uses
the definition of \(\mu\).  This proves
\eqref{eq:app-skeleton-excess}.

Because \(B(K)\) is connected and contains a cycle, \(\mu\geq1\).
Also, \(B(K)\) has at least \(s+1\) vertices and at most \(rs\)
edges, so
\[
    \mu
    =
    |E(B(K))|-|V(B(K))|+1
    \leq
    (r-1)s.
\]
Every term in the sum on the right-hand side of
\eqref{eq:app-skeleton-weight} is nonnegative.  Therefore
\(
    w\geq qs+2-2\mu,
\)
which gives
\[
    s\leq\frac{w+2\mu-2}{q}.
\]

Finally, suppose that \(\mu=1\).  The left-hand side of
\eqref{eq:app-skeleton-excess} is then zero.  Since every term in
that expression is nonnegative, every check node and every variable node in
\(B(K)\) has degree two.  As \(B(K)\) is connected, it is an
alternating cycle.  Equation~\eqref{eq:app-skeleton-weight} then
reduces to
\(
    w=qs.
\)
The same argument applies in the configuration model with edge
multiplicity; there the cycle may consist of two parallel edges,
a possibility excluded when the Tanner graph is simple.
\end{proof}

We now prove the two probabilistic estimates used in
Section~\ref{sec:proof-locality-barrier}-A for the two possible
ranges of the cyclomatic number \(\mu=1\) and \(\mu\geq 2\).
Throughout the next two proofs,
\(C<\infty\) denotes a positive constant depending only on the fixed
parameters in the lemma under consideration; its value may change
from line to line.

\begin{proof}[Proof of Lemma~\ref{lem:gated-active-zone}]
Write \(W:=W_n\).  We first work in the configuration model with
independent BEC erasures.  In this model, \(R_0(E_\epsilon)\) denotes
the terminal peeling residual with check degrees counted by incidences
with multiplicity, and \(X_{\mathrm{cyc}}^{(\epsilon)}(W)\) is
interpreted with the same multigraph convention.  On the simplicity
event these agree with the corresponding quantities in the uniform
simple ensemble.

By Theorem~\ref{thm:skeleton-identity}, every connected cyclic
skeleton with \(\mu=1\) is an alternating cycle.  If such a skeleton
contains \(s\) checks, then its generated row has weight \(qs\).
Since only rows of weight at most \(W\) contribute to
\(X_{\mathrm{cyc}}^{(\epsilon)}(W)\), in the configuration model we
need only consider
\(
    1\leq s\leq \lfloor W/q\rfloor.
\)
The case \(s=1\) is the two-parallel-edge cycle permitted by the
configuration model and disappears after conditioning on simplicity.
Since \(W=c\log n+O(1)\), every value of \(s\) considered below is
\(O(\log n)\).

Fix such an integer \(s\).  A prescribed cycle \(\Gamma\) consists
of \(s\) distinct check nodes and \(s\) distinct variable nodes in
alternating cyclic order, together with the two half-edges used by
the cycle at each node.  Let \(Z_s^\circ\) be the number of
prescribed cycles of this form that occur in the configuration.
We first compute \(\E Z_s^\circ\).  There are
\(
    (m)_s(n)_s/(2s)
\)
alternating cyclic arrangements of the node labels.  Indeed, ordered
lists of the \(s\) check labels and the \(s\) variable labels specify
a starting point and an orientation, and each unoriented cycle is
represented \(2s\) times.  At every check node the two cycle
incidences use an ordered pair of distinct half-edges, giving
\(r(r-1)\) choices, while at every variable node they give
\(\ell(\ell-1)\) choices.  Thus the number of half-edge-specified
candidate cycles is
\[
    \frac{(m)_s(n)_s}{2s}
    [r(r-1)\ell(\ell-1)]^s.
\]
For any one such candidate, the \(2s\) required check-to-variable
pairings occur in a uniformly random configuration with probability
\(
    1/(\ell n)_{2s}.
\)
Therefore, since \(m=\ell n/r\) and \(s=O(\log n)\), the standard
falling-factorial estimate gives, uniformly in the range above,
\begin{equation}
    \E Z_s^\circ
    =
    \frac{(m)_s(n)_s[r(r-1)\ell(\ell-1)]^s}
         {2s(\ell n)_{2s}}
    \leq
    C\frac{b^s}{s}.
    \label{eq:app-cycle-count}
\end{equation}
The count above may also include cycles whose check nodes have
additional shared incidences outside the prescribed cycle, so not
every counted cycle is itself a connected cyclic skeleton.  This
only enlarges the counted family and is harmless for an upper bound.

We next bound the probability that all variable nodes of a fixed
prescribed cycle \(\Gamma\) survive the original peeling process.
To express a finite-depth necessary condition for this event, index
the usual BEC messages by incidences.  For an incidence \(a\), let
\(v(a)\) and \(f(a)\) denote its variable and check endpoints, and
write \(I(x)\) for the set of incidences at a node \(x\).  Set
\[
    U_a^{(0)}=e_{v(a)},
\]
with updates
\begin{align*}
    V_a^{(t)}
    &=
    1-\prod_{a'\in I(f(a))\setminus\{a\}}
    \bigl(1-U_{a'}^{(t)}\bigr),\\
    U_a^{(t+1)}
    &=
    e_{v(a)}
    \prod_{a'\in I(v(a))\setminus\{a\}}
    V_{a'}^{(t)}.
\end{align*}
Here \(e_v\) is the channel-erasure indicator, and a message equal to
one is an erasure message.
Suppose that every variable node of \(\Gamma\) belongs to the
terminal residual \(R_0(E_\epsilon)\).  Each such variable is erased.
Consider any one of its \(\ell-2\) incidences not used by the
prescribed cycle, and let \(f\) be the check at its other end.
Since \(R_0(E_\epsilon)\) is a stopping set, \(f\) has residual
degree at least two, with incidences counted with multiplicity.
Thus, besides the incidence under consideration, \(f\) has at least
one other incidence whose variable endpoint belongs to
\(R_0(E_\epsilon)\).  Consequently, the check-to-variable message
along the original incidence remains equal to one at every finite
iteration.  The same holds for each of the \(\ell-2\) incidences
outside the prescribed cycle.
It follows that, for every fixed \(T\), membership of all cycle
variables in \(R_0(E_\epsilon)\) implies the following finite-depth
condition: every cycle variable is erased and all of its
\(\ell-2\) incoming messages through incidences not used by
\(\Gamma\) are equal to one at iteration \(T\).

We now estimate the probability of this necessary condition by exposing
the depth-\(T\) computation neighborhoods that determine these
messages.  The exploration reveals \(O(sb^T)\) half-edges in total.
A collision occurs if an exposed branch revisits a Tanner-graph node,
meets a branch grown from another cycle variable, or returns to a
node of \(\Gamma\).
Consider the required pairings in the order in which they are
exposed.  Conditional on all pairings revealed so far, the remaining
unmatched half-edges are still uniformly paired.  At the \(j\)-th
exposure, only \(O(s+j)\) partner half-edges can create a collision,
whereas the number of available partner half-edges is of order
\(n\).  Thus the conditional collision probability at that exposure
is
\(
    O\!\left(\frac{s+j}{n}\right).
\)
If \(M_T=O(sb^T)\) denotes the total number of exposures, a union
bound then gives
\begin{equation}
     \Prob\{
        \text{collision}\mid\Gamma\text{ occurs}
    \}
    \leq
    \sum_{j=1}^{M_T}
    O\!\left(\frac{s+j}{n}\right)
    =
    O\!\left(
        \frac{sM_T+M_T^2}{n}
    \right)
    \leq
    C\frac{s^2b^{2T}}{n}.
    \label{eq:app-cycle-collision}
\end{equation}

Suppose now that no collision occurs.  The depth-\(T\) computation
neighborhoods associated with the \(s\) cycle variables are then
disjoint regular trees.  They involve disjoint sets of channel
variables, and the BEC erasures on these sets are independent.
Moreover, the erasure pattern is independent of the random matching,
so conditioning on the collision-free matching event does not change
the distribution of the channel erasures.
By the density-evolution recursion of Section~\ref{sec:model}, a
check-to-variable message at iteration \(T\) is erased with
probability \(y_T\).  For one cycle variable, the finite-depth
condition requires its own channel symbol to be erased, with
probability \(\epsilon\), and each of its \(\ell-2\) external
incoming messages to be erased, each with probability \(y_T\).
These events depend on disjoint branches.  Thus one cycle variable
satisfies the finite-depth condition with probability
\(
    p_T:=\epsilon y_T^{\ell-2}.
\)
Because the computation neighborhoods rooted at distinct cycle
variables are disjoint, all \(s\) cycle variables satisfy the
condition simultaneously with probability \(p_T^s\).
If all cycle variables belong to \(R_0(E_\epsilon)\), then the
finite-depth condition must hold.  On the collision-free event its
probability is \(p_T^s\), while on the collision event we use the
trivial upper bound one.  Combining these two cases with
\eqref{eq:app-cycle-collision} gives
\begin{equation}
    \Prob\{
        \text{all variable nodes of }\Gamma
        \text{ belong to }R_0(E_\epsilon)
        \mid \Gamma\text{ occurs}
    \}
    \leq
    p_T^s
    +
    C\frac{s^2b^{2T}}{n}.
    \label{eq:app-cycle-survival}
\end{equation}
For \(1\leq s\leq\lfloor W/q\rfloor\), let \(N_s\) be the number of
cycles counted by \(Z_s^\circ\) whose variable nodes all belong to
\(R_0(E_\epsilon)\).  By \eqref{eq:app-cycle-survival},
\begin{align*}
    \E N_s
    &=
    \sum_{\Gamma}
    \Prob\{
        \Gamma\text{ occurs and all of its variable nodes belong to }
        R_0(E_\epsilon)
    \}\\
    &\leq
    \left(
        p_T^s
        +
        C\frac{s^2b^{2T}}{n}
    \right)
    \sum_{\Gamma}\Prob\{\Gamma\text{ occurs}\}\\
    &=
    \E Z_s^\circ
    \left(
        p_T^s
        +
        C\frac{s^2b^{2T}}{n}
    \right),
\end{align*}
where the sums run over all prescribed cycles \(\Gamma\) with
\(s\) check nodes and \(s\) variable nodes.

Every connected cyclic skeleton contributing to
\(X_{\mathrm{cyc}}^{(\epsilon)}(W)\) with \(s\) checks determines one
of the cycles counted by \(N_s\), and by
Theorem~\ref{thm:skeleton-identity} its support has size exactly
\(qs\).  The cardinality of a union is at most the sum of the
cardinalities of its constituent supports.  Hence
\[
    |X_{\mathrm{cyc}}^{(\epsilon)}(W)|
    \leq
    \sum_{s=1}^{\lfloor W/q\rfloor} qs\,N_s.
\]
Taking expectations, applying the preceding bound on \(\E N_s\), and
then using \eqref{eq:app-cycle-count}, we obtain
\begin{equation}
    \E|X_{\mathrm{cyc}}^{(\epsilon)}(W)|
    \leq
    C\sum_{s=1}^{\lfloor W/q\rfloor}(bp_T)^s
    +
    \frac{Cb^{2T}}{n}
    \sum_{s=1}^{\lfloor W/q\rfloor}s^2b^s.
    \label{eq:app-finite-depth-cycles}
\end{equation}
Conditioning on simplicity changes the preceding bound only by a
constant factor, so \eqref{eq:app-finite-depth-cycles} also holds in
the uniform simple ensemble after adjusting \(C\).

It remains to choose the depth \(T\).  By the upper-fixed-point
stability recalled in Section~\ref{sec:model},
\(f_\epsilon'(x)<1\).  Since \(f_\epsilon\) is continuously
differentiable and \(x_t\to x\), the convergence is geometric once
\(x_t\) enters a sufficiently small neighborhood of \(x\).
Hence, with
\(p=\epsilon y^{\ell-2}=x/y\), there are
\(\kappa\in(0,1)\) and \(C<\infty\) such that
\[
    |p_T-p|\leq C\kappa^T.
\]
Choose
\(
    T=\lceil \alpha\log\log n\rceil
\)
with \(\alpha\) large enough that
\(
    C\kappa^T=O((\log n)^{-2}).
\)
Then
\(
    |p_T-p|=O((\log n)^{-2})
\)
and
\(
    b^{2T}=(\log n)^{O(1)}.
\)
Since
\(y=1-(1-x)^{r-1}<(r-1)x\), we have
\(bp>\ell-1>1\).  Thus \(bp_T>1\) for all sufficiently large \(n\).
As \(W=c\log n+O(1)\), the two sums in
\eqref{eq:app-finite-depth-cycles} satisfy
\[
    \sum_{s=1}^{\lfloor W/q\rfloor}(bp_T)^s
    \leq
    n^{(c/q)\log(bp)+o(1)},
    \qquad
    \frac{b^{2T}}{n}
    \sum_{s=1}^{\lfloor W/q\rfloor}s^2b^s
    \leq
    n^{(c/q)\log b-1+o(1)}.
\]
The first bound follows from \(bp_T>1\) and
\(\log(bp_T)=\log(bp)+o(1/\log n)\).
For the second sum, since \(b>1\),
\(
    \sum_{s=1}^{\lfloor W/q\rfloor}s^2b^s
    \leq
    C W^2 b^{W/q}.
\)
Because \(W=O(\log n)\) and \(b^{2T}=(\log n)^{O(1)}\), the
polynomial factors are \(n^{o(1)}\), while
\(b^{W/q}=n^{(c/q)\log b+o(1)}\).

Using the definitions
\(c_{\mathrm{gate}}(\epsilon)=q/\log(bp)\) and
\(c_0=q/\log b\), \eqref{eq:app-finite-depth-cycles} now gives
\[
    \E|X_{\mathrm{cyc}}^{(\epsilon)}(W_n)|
    \leq
    n^{c/c_{\mathrm{gate}}(\epsilon)+o(1)}
    +
    n^{c/c_0-1+o(1)}.
\]
This proves \eqref{eq:gated-cycle-asymptotic}.  If
\(c<\min\{c_{\mathrm{gate}}(\epsilon),2c_0\}\), the right-hand side is
\(o(n)\), and the convergence in probability follows immediately
from Markov's inequality.
\end{proof}

\begin{proof}[Proof of Lemma~\ref{lem:higher-excess-support}]
Write \(W:=W_n\) and work first in the configuration model.
Every skeleton contributing to \(X_{\geq2}(W)\) has at most \(W\)
checks, so throughout this proof
\(
    s=O(\log n).
\)

For \(s\geq1\) and \(\mu\geq2\), let \(Z_{s,\mu}\) denote the
number of connected cyclic skeletons with \(s\) checks and
cyclomatic number \(\mu\).  Let \(\mathcal F_{s,\mu}\) be the
collection of isomorphism types of connected bipartite multigraphs
with \(s\) check-side vertices, minimum degree at least two, and
cyclomatic number \(\mu\), subject to the degree bounds \(r\) on the
check side and \(\ell\) on the variable side.
For \(F\in\mathcal F_{s,\mu}\), choose an arbitrary labeling of its
vertices and edges, and let \(N(F)\) be the number of injective
half-edge embeddings of this labeled graph in the configuration
model.  Every skeleton counted by \(Z_{s,\mu}\) gives at least one
such embedding of its abstract shared-incidence graph, and hence
\[
    Z_{s,\mu}
    \leq
    \sum_{F\in\mathcal F_{s,\mu}} N(F).
\]

Fix \(F\in\mathcal F_{s,\mu}\), and let \(p\) and \(e\) denote its
numbers of variable-side vertices and edges, respectively.
As in the cycle count in the proof of
Lemma~\ref{lem:gated-active-zone}, first choose distinct check nodes
and variable nodes for the vertices of \(F\), and then choose the
half-edges used by its edges.  The number of half-edge-specified
candidates is
\[
    (m)_s(n)_p
    \prod_{f\in F}(r)_{d_F(f)}
    \prod_{v\in F}(\ell)_{d_F(v)}.
\]
For each candidate, the \(e\) prescribed pairings occur with
probability \(1/(\ell n)_e\).  Therefore
\[
    \E N(F)
    =
    \frac{(m)_s(n)_p}{(\ell n)_e}
    \prod_{f\in F}(r)_{d_F(f)}
    \prod_{v\in F}(\ell)_{d_F(v)}.
\]
The event counted by \(N(F)\) requires the prescribed edges of
\(F\) to occur but does not require them to be all the shared
incidences among the chosen nodes.  Thus \(N(F)\) may count
embeddings whose full shared-incidence graph is larger than \(F\),
which is harmless in the upper bound above.

To estimate this expression, put
\(
    \Delta_{\mathrm c}:=e-2s
\)
and
\(
    \Delta_{\mathrm v}:=e-2p.
\)
Since every vertex of \(F\) has degree at least two, both quantities
are nonnegative.  From
\(
    \mu=e-s-p+1
\)
we have
\(
    \Delta_{\mathrm c}+\Delta_{\mathrm v}=2\mu-2
\)
and
\(
    2(p-s)=\Delta_{\mathrm c}-\Delta_{\mathrm v}.
\)
In particular,
\(
    |p-s|\leq (\Delta_{\mathrm c}+\Delta_{\mathrm v})/2 = \mu-1.
\)
Only graphs \(F\) whose check-side degrees are at most \(r\) and
variable-side degrees are at most \(\ell\) can occur.  Hence
\(e\leq rs=O(\log n)\), and since every variable-side vertex has
degree at least two, also \(p=O(\log n)\).  Therefore the
falling-factorial terms in the preceding embedding count satisfy
\[
    \frac{(m)_s(n)_p}{(\ell n)_e}
    =
    \frac{m^s n^p}{(\ell n)^e}e^{o(1)}
\]
uniformly over the range considered here.
Since every degree in \(F\) is at least two,
\[
    \prod_{f\in F}(r)_{d_F(f)}
    \leq
    [r(r-1)]^s r^{\Delta_{\mathrm c}},
    \qquad
    \prod_{v\in F}(\ell)_{d_F(v)}
    \leq
    [\ell(\ell-1)]^p \ell^{\Delta_{\mathrm v}}.
\]
Substituting these bounds and \(m=\ell n/r\) into the embedding
count gives
\[
\begin{aligned}
    \E N(F)
    &\leq
    e^{o(1)}
    \frac{(\ell n/r)^s n^p}{(\ell n)^e}
    [r(r-1)]^s r^{\Delta_{\mathrm c}}
    [\ell(\ell-1)]^p\ell^{\Delta_{\mathrm v}}\\
    &=
    e^{o(1)}
    n^{1-\mu}b^s
    r^{\Delta_{\mathrm c}}
    \left(\frac{\ell-1}{\ell}\right)^{p-s}.
\end{aligned}
\]
Here we used \(s+p-e=1-\mu\) and
\(\Delta_{\mathrm v}=e-2p\).
Finally,
\(
    \Delta_{\mathrm c}\leq2\mu-2
\)
and
\(
    |p-s|\leq\mu-1
\),
so the last two factors are bounded by \(C^\mu\).  Thus
\[
    \E N(F)
    \leq
    n^{1-\mu}b^sC^\mu.
\]

It remains to bound \(|\mathcal F_{s,\mu}|\). The point is that, after
suppressing degree-two paths, the remaining branching structure has
size depending only on \(\mu\).
Suppress every maximal path whose internal vertices have degree two,
and let \(K_F\) be the resulting kernel.  Since \(\mu\geq2\),
\(K_F\) is nonempty and every vertex of \(K_F\) has degree at least
three.  Moreover,
\(
    \sum_{x\in V(F)}(d_F(x)-2)=2\mu-2.
\)
Suppression does not change this total excess or the cyclomatic
number.  Hence \(K_F\) has at most \(2\mu-2\) vertices and, since it
is connected with cyclomatic number \(\mu\), at most \(3\mu-3\)
edges.

We use only a crude upper bound on the number of possible kernels.
Temporarily number their vertices and edges, record the bipartite
type of each vertex, and record the two, not necessarily distinct,
endpoints of every edge.  Since the kernel has \(O(\mu)\) vertices
and edges, the number of such descriptions is at most
\(
    (C\mu)^{6\mu}.
\)
Now fix a kernel with \(k\leq3\mu-3\) edges.  The original graph
\(F\) is obtained by replacing each kernel edge by an alternating
degree-two path.  Record for each kernel edge the number of internal
check-side vertices on that path.  Since the total number of checks
is \(s\), the number of possible choices is at most
\(
    \binom{s+k}{k}\leq(Cs)^k.
\)
Once these numbers and the endpoint types are fixed, the numbers of
internal variable-side vertices are determined by alternation.
Thus, for fixed \(s\) and \(\mu\), the number of graphs that need to
be considered is at most
\(
    (C\mu)^{6\mu}(Cs)^{3\mu}.
\)
Combining the fixed-\(F\) estimate with the preceding bound on
\(|\mathcal F_{s,\mu}|\), we obtain
\begin{equation}
\begin{aligned}
    \E Z_{s,\mu}
    &\leq
    \sum_{F\in\mathcal F_{s,\mu}}\E N(F)\\
    &\leq
    n^{1-\mu}b^s
    C^\mu(C\mu)^{6\mu}(Cs)^{3\mu}\\
    &\leq
    n^{1-\mu}b^s(Cs)^{12\mu}.
\end{aligned}
\label{eq:app-skeleton-count}
\end{equation}

We now use \eqref{eq:app-skeleton-count} to bound the support union.
Each skeleton contributing to \(X_{\geq2}(W)\) has
\(\wt(h_K)\leq W\), and hence contributes at most \(W\) coordinates
to the union.  Such a skeleton also has \(s\leq W\) by definition.
Moreover, Theorem~\ref{thm:skeleton-identity} gives
\[
    s\leq\frac{W+2\mu-2}{q},
    \qquad
    \mu\leq(r-1)s\leq(r-1)W.
\]
Therefore, by the union bound and
\eqref{eq:app-skeleton-count},
\[
\begin{aligned}
    \E|X_{\geq2}(W)|
    &\leq
    W
    \sum_{\mu=2}^{(r-1)W}
    \sum_{s\leq\min\{W,(W+2\mu-2)/q\}}
    \E Z_{s,\mu}\\
    &\leq
    W
    \sum_{\mu=2}^{(r-1)W}
    n^{1-\mu}(CW)^{12\mu}
    \sum_{s\leq\min\{W,(W+2\mu-2)/q\}} b^s.
\end{aligned}
\]
For fixed \(\mu\), the last sum is geometric.  Since \(b>1\),
\[
    \sum_{s\leq\min\{W,(W+2\mu-2)/q\}} b^s
    \leq
    C b^{(W+2\mu-2)/q}
    =
    C b^{W/q}b^{2(\mu-1)/q}.
\]
Hence
\[
    \E|X_{\geq2}(W)|
    \leq
    C W b^{W/q}
    \sum_{\mu=2}^{(r-1)W}
    n^{1-\mu}
    b^{2(\mu-1)/q}
    (CW)^{12\mu}.
\]
The summands in the \(\mu\)-sum form a geometric progression with
common ratio
\(
    \rho_n
    :=
    n^{-1}b^{2/q}(CW)^{12}.
\)
Since \(W=O(\log n)\), we have \(\rho_n=o(1)\).  In particular,
\(\rho_n\leq1/2\) for all sufficiently large \(n\), so the geometric
sum is at most twice its first term, corresponding to \(\mu=2\).
Hence
\begin{equation}
    \E|X_{\geq2}(W)|
    \leq
    C n^{-1}b^{W/q}W^C.
    \label{eq:app-higher-support}
\end{equation}
As in the proof of Lemma~\ref{lem:gated-active-zone}, conditioning
on simplicity changes this bound only by a constant factor.
Therefore \eqref{eq:app-higher-support} also holds in the uniform
simple ensemble.

Finally, as \(W_n=c\log n+O(1)\),
\(
    b^{W_n/q}
    =
    n^{(c/q)\log b+o(1)}
    =
    n^{c/c_0+o(1)},
\)
while \(W_n^{O(1)}=n^{o(1)}\).  Therefore
\[
    \E|X_{\geq2}(W_n)|
    \leq
    n^{c/c_0-1+o(1)}.
\]
If \(c<2c_0\), then for every fixed \(\delta>0\), Markov's inequality gives
\[
    \Prob\!\left\{
        |X_{\geq2}(W_n)|>\delta n
    \right\}
    \leq
    \frac{\E|X_{\geq2}(W_n)|}{\delta n}
    =
    n^{c/c_0-2+o(1)}
    \frac1\delta
    \longrightarrow0.
\]
This proves the lemma.
\end{proof}

\subsection{Residual stability under deletions}
\label{app:residual-stability}

This appendix proves Theorem~\ref{thm:core-resilience}.

\begin{proof}[Proof of Theorem~\ref{thm:core-resilience}]
Fix \(\epsilon\in(\epsBP,1]\) and write
\(
    R:=R_0(E_\epsilon)
\)
and
\(
    N:=|R|.
\)
As in Appendix~\ref{app:dual-checks}, we first work in the bipartite
configuration model, with incidences counted with multiplicity.  Let
\((x,y)\) be the nonzero fixed point reached by the density-evolution
recursion initialized at \(x_0=\epsilon\).

For \(2\leq d\leq r\), let \(M_d\) be the number of check nodes having
exactly \(d\) incidences with the variable nodes in \(R\).
The terminal residual degree profile is given in
Eqs.~(9)--(10) of M\'easson, Montanari, and Urbanke~\cite{MeassonMontanariUrbanke2008}, 
with concentration at continuity points in their Lemma~6.
For the regular ensemble considered here, the upper nonzero
fixed-point branch is continuous for
\(\epsilon\in(\epsBP,1)\), as recalled in
Section~\ref{sec:model}.  Hence the profile appearing in those
equations is continuous throughout this interval, and their
Lemma~6 applies.  Specializing to the \((\ell,r)\)-regular ensemble
gives jointly
\begin{equation}
 \frac{N}{n}\xrightarrow{\Prob}\rho:=xy>0,
 \qquad
 \frac{M_d}{n}\xrightarrow{\Prob}
   \frac{\ell}{r}\binom rd x^d(1-x)^{r-d},
 \quad 2\leq d\leq r.
 \label{eq:app-residual-profile}
\end{equation}
For \(\epsilon=1\), the same relations hold directly:
\(R=V_n\), \(M_r=m\), and \(M_d=0\) for \(2\leq d<r\).

Construct a multigraph \(G_2\) with vertex set \(R\) as follows.
For each check node counted by \(M_2\), let \(u,v\in R\) be the
variable endpoints of its two incidences with \(R\), and add one edge
\(uv\) to \(G_2\).  Loops and parallel edges are allowed in the
configuration model.  Let \(\mathcal C_2\) denote the connected
components of \(G_2\).
If one variable in a component of \(G_2\) is removed, the check nodes
represented by the edges of that component allow the remaining
variables in the component to be peeled successively.  We therefore
need to control the sizes of these components.

M\'easson, Montanari, and Urbanke state that the
terminal residual graph is uniformly distributed conditional on its
degree profile in Section~V of~\cite{MeassonMontanariUrbanke2008}.
In the configuration-model formulation, after conditioning further on
the labeled residual set \(R\) and on the residual degree of each
check node, the compatible residual half-edge pairings are therefore
uniform.  Consequently, the \(2M_2\) variable-side half-edges paired
to the check nodes counted by \(M_2\) form a uniformly chosen subset
of size \(2M_2\) among the \(\ell N\) variable-side half-edges
incident to \(R\).
Write \(\theta_n:=2M_2/(\ell N)\).  For \(v\in R\), let \(K_v\) be
its degree in \(G_2\).  Since \(K_v\) counts how many of the \(\ell\)
half-edges at \(v\) belong to the uniformly chosen subset above, its
conditional distribution is hypergeometric and, since \(\ell\) is
fixed, is asymptotically \(\operatorname{Bin}(\ell,\theta_n)\).
Standard concentration for sampling without replacement shows that the
empirical degree distribution of \(G_2\) is therefore asymptotically
\(\operatorname{Bin}(\ell,\theta_n)\).  By \eqref{eq:app-residual-profile},
\(\theta_n\to\theta:=(r-1)x(1-x)^{r-2}/y\) in probability, so this
distribution converges to \(\operatorname{Bin}(\ell,\theta)\).  In an
exploration of a connected component of \(G_2\), every variable after
the first is reached along one edge of \(G_2\); applying the same
sampling argument to its remaining \(\ell-1\) half-edges shows that
the number of further edges incident to that variable has limiting
distribution \(\operatorname{Bin}(\ell-1,\theta)\), with mean
\[
    (\ell-1)\theta=f_\epsilon'(x)<1,
\]
where the strict inequality is the upper-fixed-point stability
recalled in Section~\ref{sec:model}.

We now bound the total number of variables contained in an arbitrary
collection of connected components of \(G_2\).
If \(\epsilon=1\), then \(M_2=0\) and \(G_2\) has no edges, so the
component bounds below are immediate.
Suppose now that \(\epsilon<1\).  For a fixed degree sequence
\((K_v)_{v\in R}\), the mean number of other \(G_2\)-edges incident
to the variable endpoint of a uniformly chosen \(G_2\)-half-edge is
\[
    \nu_n
    :=
    \frac{\sum_{v\in R}K_v(K_v-1)}
         {\sum_{v\in R}K_v}.
\]
Since the degrees are bounded and their empirical distribution
converges to \(\operatorname{Bin}(\ell,\theta)\), while
\(\theta>0\) for \(\epsilon<1\),
\(
    \nu_n
    \xrightarrow{\Prob}
    (\ell-1)\theta
    <1.
\)
Hence there are constants \(a_0,\delta>0\) such that, with
probability \(1-o(1)\),
\(
    \sum_{v\in R}K_v=2M_2\geq a_0 n,
\)
\(
    K_v\leq\ell
\)
for every \(v\in R\), and
\(
    \nu_n\leq1-\delta.
\)
Indeed, the first bound follows from
\eqref{eq:app-residual-profile}, since \(M_2/n\) converges to a
strictly positive constant when \(\epsilon<1\).
Fix a degree sequence satisfying these three bounds.  Conditional on
this degree sequence, \(G_2\) is a uniform configuration multigraph
with degrees \((K_v)_{v\in R}\).  Until the conditioning is removed
below, probabilities and expectations are taken over this uniform
pairing.

Let \(K_{\mathrm{tail}}<\infty\) be arbitrary and explore the
component of a vertex \(v\) by breadth first search.  Before
\(K_{\mathrm{tail}}\log n\) vertices have been exposed, only
\(O(\log n)\) half-edges have been removed.  Because
\(K_v\leq\ell\) and
\(
    \sum_{v\in R}K_v\geq a_0n,
\)
the conditional mean number of new half-edges revealed when an
unexplored half-edge is followed differs from \(\nu_n\) by
\(O(\log n/n)\), uniformly over the exploration history.
Since \(\nu_n\leq1-\delta\), this conditional mean is therefore at
most \(1-\delta/2\) for all sufficiently large \(n\).
Let \(X\) denote the number of new half-edges revealed in one such
step.  Then \(0\leq X\leq\ell-1\).  The bounded support and the
uniform gap below one imply that there are constants
\(\lambda,\kappa>0\), depending only on \((\ell,\delta)\), such that
\[
    \E\!\left[
        e^{\lambda(X-1)}
        \,\middle|\,
        \text{exploration history}
    \right]
    \leq e^{-\kappa}.
\]
Collisions during the exploration can only decrease the number of
active half-edges.  Since \(\lambda\) and \(\kappa\) do not depend on
\(K_{\mathrm{tail}}\), the resulting exponential-supermartingale
argument gives constants \(A_{\mathrm{tail}},\tau>0\), depending
only on \((\ell,\delta)\), such that, for every fixed
\(K_{\mathrm{tail}}<\infty\), all sufficiently large \(n\), and every
\(v\in R\),
\[
    \Prob\{|B(v)|\geq t\}
    \leq
    A_{\mathrm{tail}}e^{-\tau t},
    \qquad
    1\leq t\leq K_{\mathrm{tail}}\log n.
\]
The root contributes at most \(\ell\) initial active half-edges,
which is absorbed into \(A_{\mathrm{tail}}\).
See, for example, \cite[Sec.~5.1]{IbrahimiEtAl2015} for the standard
subcritical breadth-first-search comparison.
We may now choose \(K_{\mathrm{tail}}>2/\tau\) and put
\(
    L:=\lfloor K_{\mathrm{tail}}\log n\rfloor.
\)
A union bound over the variables in \(R\) gives
\[
\begin{aligned}
    \Prob\!\left\{
        \max_{B\in\mathcal C_2}|B|\geq L
    \right\}
    &\leq
    \sum_{v\in R}\Prob\{|B(v)|\geq L\}\\
    &\leq
    nA_{\mathrm{tail}}e^{-\tau L}\\
    &\leq
    A_{\mathrm{tail}}e^\tau n^{1-\tau K_{\mathrm{tail}}}
    =
    o(1),
\end{aligned}
\]
Thus every component of
\(G_2\) has fewer than \(L\) variables with probability \(1-o(1)\).

Choose \(a\) with
\(
    0<a<1/(2K_{\mathrm{tail}}).
\)
Since \(K_{\mathrm{tail}}>2/\tau\), we also have \(a<\tau\).  Define
\[
    F_L
    :=
    \sum_{B\in\mathcal C_2}
    e^{a\min\{|B|,L\}}.
\]
For every \(v\in R\), the preceding exponential tail gives
\[
\begin{aligned}
    \E e^{a\min\{|B(v)|,L\}}
    &=
    1+(e^a-1)
    \sum_{t=1}^{L}
        e^{a(t-1)}
        \Prob\{|B(v)|\geq t\}\\
    &\leq
    1+A_{\mathrm{tail}}(e^a-1)
    \sum_{t=1}^{\infty}
        e^{a(t-1)}e^{-\tau t}\\
    &=
    1+A_{\mathrm{tail}}(1-e^{-a})
        \sum_{t=1}^{\infty}e^{-(\tau-a)t}\\
    &=
    1+
    A_{\mathrm{tail}}(1-e^{-a})
    \frac{e^{-(\tau-a)}}{1-e^{-(\tau-a)}}\\
    &\leq
    M_{\mathrm{exp}},
\end{aligned}
\]
for a constant \(M_{\mathrm{exp}}<\infty\).  Also,
\[
    F_L
    =
    \sum_{B\in\mathcal C_2}
        e^{a\min\{|B|,L\}}
    \leq
    \sum_{v\in R}
        e^{a\min\{|B(v)|,L\}},
\]
because a component \(B\) contributes one term to the first sum and
\(|B|\) identical terms to the second.  Therefore
\[
    \E F_L
    \leq
    M_{\mathrm{exp}}N
    \leq
    M_{\mathrm{exp}}n.
\]

Expose the uniform pairing defining \(G_2\) one pair at a time.
Consider two possible choices for the partner of the half-edge being
exposed.  The two resulting completions can be coupled so that they
differ only by a switching of two pairs.  After those two pairs are
removed, the affected part consists of at most four connected
pieces.  Hence, in either completion, at most four component terms
can be affected, and the two values of \(F_L\) differ by at most
\(
    8e^{aL}.
\)
There are \(M_2\leq\ell n/2\) exposure steps.  The
Azuma--Hoeffding inequality therefore gives, for every fixed \(s>0\),
\[
\begin{aligned}
    \Prob\!\left\{
        F_L-\E F_L\geq sn
    \right\}
    &\leq
    \exp\!\left\{
        -\frac{s^2n^2}
        {2M_2(8e^{aL})^2}
    \right\}\\
    &\leq
    \exp\!\left\{
        -\xi_1n e^{-2aL}
    \right\}\\
    &\leq
    \exp\!\left\{
        -\xi_2n^{\,1-2aK_{\mathrm{tail}}}
    \right\},
\end{aligned}
\]
for constants \(\xi_1,\xi_2>0\).  Since \(2aK_{\mathrm{tail}}<1\), the
last expression is \(o(1)\).  Taking \(s=1\) and using \(\E F_L\leq M_{\mathrm{exp}}n\), we obtain
\(
    F_L\leq (M_{\mathrm{exp}}+1)n
\)
with probability \(1-o(1)\), uniformly over all degree sequences in
the high-probability set fixed above.

The degree sequence belongs to that set with probability \(1-o(1)\).
Removing the conditioning, and combining the last estimate with the
largest-component bound, there is therefore a constant \(M_{\mathrm{comp}}<\infty\)
such that, with probability \(1-o(1)\),
\[
    \sum_{B\in\mathcal C_2}e^{a|B|}
    \leq
    M_{\mathrm{comp}}n.
\]
Indeed, on the event that every component has size less than \(L\),
the truncation in \(F_L\) is inactive.

Now let \(\mathcal S\subseteq\mathcal C_2\) be any nonempty collection of
components.  Since its components are disjoint, Jensen's inequality
gives
\[
\begin{aligned}
    \exp\!\left\{
        \frac{a}{|\mathcal S|}
        \left|
            \bigcup_{B\in\mathcal S}B
        \right|
    \right\}
    &=
    \exp\!\left\{
        \frac{a}{|\mathcal S|}
        \sum_{B\in\mathcal S}|B|
    \right\}\\
    &\leq
    \frac{1}{|\mathcal S|}
    \sum_{B\in\mathcal S}e^{a|B|}\\
    &\leq
    \frac{M_{\mathrm{comp}}n}{|\mathcal S|}.
\end{aligned}
\]
Taking logarithms yields, simultaneously for every nonempty
\(\mathcal S\subseteq\mathcal C_2\),
\begin{equation}
    \left|
        \bigcup_{B\in\mathcal S}B
    \right|
    \leq
    \frac{|\mathcal S|}{a}
    \log\frac{M_{\mathrm{comp}}n}{|\mathcal S|}.
    \label{eq:app-block-union}
\end{equation}
For \(\epsilon=1\), \(G_2\) has no edges, so
\eqref{eq:app-block-union} is immediate.  Thus the bound holds with
probability \(1-o(1)\) for every fixed
\(\epsilon\in(\epsBP,1]\), and we treat both cases together below.

We next consider a union of \(k\) components of \(G_2\) and bound
the probability that many check nodes having at least three
incidences with \(R\) have all the variable endpoints of those
incidences in that union.  The pairings involving the degree-two check nodes have
already been exposed.  Conditional uniformity therefore implies that
the remaining unmatched half-edges are still uniformly matched.
Let
\(
    L_+:=\ell N-2M_2
\)
be the number of remaining variable-side half-edges.  Since
\(\theta_n=2M_2/(\ell N)\), we have
\[
    \frac{L_+}{n}
    =
    \ell\frac{N}{n}(1-\theta_n)
    \xrightarrow{\Prob}
    \ell\rho(1-\theta)>0.
\]
Thus there is a constant \(c_+>0\) such that
\(
    L_+\geq c_+n
\)
with probability \(1-o(1)\).
Fix an integer \(k\geq1\) and put
\(
    t:=\lceil3k/4\rceil.
\)
Let \(\mathcal S\subseteq\mathcal C_2\) be any family of \(k\)
components and set
\(
    U:=\bigcup_{B\in\mathcal S}B.
\)
We call a check node \emph{internal to \(U\)} if it has at least
three incidences with \(R\) and every one of those incidences has its
variable endpoint in \(U\).
On the event in \eqref{eq:app-block-union},
\[
    |U|
    \leq
    \frac{k}{a}\log\frac{M_{\mathrm{comp}}n}{k}.
\]
Hence the number \(L_U\) of unexposed variable-side half-edges
incident to variables in \(U\) satisfies
\[
    L_U
    \leq
    \ell|U|
    \leq
    \frac{\ell k}{a}\log\frac{M_{\mathrm{comp}}n}{k}.
\]

Consider \(t\) specified check nodes,
each having at least three incidences with \(R\), and let \(D_0\) be
the total number of their incidences with \(R\).  Then \(D_0\geq3t\).
Fix a realization of \(G_2\) for which
\eqref{eq:app-block-union} holds and \(L_+\geq c_+n\).
Define
\[
    C_{\mathrm{int}}:=\frac{\ell}{a c_+}.
\]
Conditional on this exposed information, the remaining matching is
uniform.  Hence
\[
\begin{aligned}
    \Prob\{
        \text{all \(t\) checks are internal to \(U\)}
    \}
    &\leq
    \frac{(L_U)_{D_0}}{(L_+)_{D_0}}\\
    &=
    \prod_{i=0}^{D_0-1}
        \frac{L_U-i}{L_+-i}\\
    &\leq
    \left(\frac{L_U}{L_+}\right)^{D_0}\\
    &\leq
    \left[
        C_{\mathrm{int}}\frac{k}{n}
        \log\frac{M_{\mathrm{comp}}n}{k}
    \right]^{D_0},
\end{aligned}
\]
There are at most \(\binom nk\) possible families of \(k\)
components and at most \(\binom mt\) choices of the \(t\) check
nodes.  Thus, whenever \(k\) satisfies
\[
    C_{\mathrm{int}}\frac{k}{n}
    \log\frac{M_{\mathrm{comp}}n}{k}\leq1,
\]
the inequality \(D_0\geq3t\) gives
\[
\begin{aligned}
 &\Prob\{
    \text{for some family of \(k\) components, at least \(t\)
    check nodes are internal to its union}
  \}\\
 &\qquad\leq
 \binom nk\binom mt
 \left[
     C_{\mathrm{int}}\frac{k}{n}\log\frac{M_{\mathrm{comp}}n}{k}
 \right]^{3t}.
\end{aligned}
\]
Using
\(
    \binom Ns\leq(eN/s)^s
\)
and \(m=\ell n/r\), the logarithm of the right-hand side, divided by
\(k\), is at most
\begin{equation}
\begin{aligned}
    &\log\frac{en}{k}
    +\frac{t}{k}\log\frac{em}{t}
    +3\frac{t}{k}
        \log\left(
            C_{\mathrm{int}}\frac{k}{n}
            \log\frac{M_{\mathrm{comp}}n}{k}
        \right)\\
    &\qquad\leq
    \left(1-2\frac{t}{k}\right)
        \log\frac{n}{k}
    +3\frac{t}{k}
        \log\log\frac{M_{\mathrm{comp}}n}{k}
    +O(1)\\
    &\qquad\leq
    -\frac12\log\frac{n}{k}
    +3\log\log\frac{M_{\mathrm{comp}}n}{k}
    +O(1).
\end{aligned}
\label{eq:app-internal-count-exponent}
\end{equation}
where the last inequality uses
\(
    3/4\leq t/k\leq1.
\)

As \(\beta\downarrow0\),
\[
    \beta\log\frac{2M_{\mathrm{comp}}}{\beta}\longrightarrow0,
    \qquad
    -\frac12\log\frac1\beta
    +3\log\log\frac{M_{\mathrm{comp}}}{\beta}
    +O(1)
    \longrightarrow-\infty.
\]
So we can choose \(\beta>0\) sufficiently small that
\[
    C_{\mathrm{int}}\beta\log\frac{2M_{\mathrm{comp}}}{\beta}<1,
    \qquad
    \frac{\beta}{a}\log\frac{2M_{\mathrm{comp}}}{\beta}<\frac{\rho}{4},
    \qquad
    \beta<\frac{M_{\mathrm{comp}}}{e}.
\]
For this choice of \(\beta\), there is a constant
\(\gamma_0>0\) such that, when \(k=\lfloor\beta n\rfloor\), the
right-hand side of \eqref{eq:app-internal-count-exponent} is at most
\(-\gamma_0\) for all sufficiently large \(n\).  Fix this value of \(\beta\), and
set
\[
    k:=\lfloor\beta n\rfloor,
    \qquad
    t:=\left\lceil\frac{3k}{4}\right\rceil.
\]
Now for all sufficiently large \(n\),
\[
    C_{\mathrm{int}}\frac{k}{n}
    \log\frac{M_{\mathrm{comp}}n}{k}<1,
\]
so the counting estimate above shows that the probability that, for
some family of \(k\) components, at least \(t\) check nodes are
internal to its union is at most
\(
    e^{-\gamma_0 k}=o(1).
\)

Combining these estimates with
\eqref{eq:app-residual-profile} and
\eqref{eq:app-block-union}, we obtain an event of probability
\(1-o(1)\) in the configuration model on which
\(N\geq3\rho n/4\), \eqref{eq:app-block-union} holds for every
nonempty \(\mathcal S\subseteq\mathcal C_2\), and, for every family
of \(k\) components, fewer than \(t\) check nodes are internal to
its union.  By the simplicity transfer established in
Appendix~\ref{app:dual-checks}, the same event has probability
\(1-o(1)\) in the uniform simple ensemble.

Fix a realization of the high-probability event above and an arbitrary
\(D\subseteq R\) with
\[
    |D|\leq\frac{\beta n}{8}.
\]
For all sufficiently large \(n\),
\(
    |D|\leq\lfloor k/4\rfloor.
\)
Define
\(
    Q
    :=
    R\setminus\Core_{H_{0,n}}(R\setminus D).
\)
Thus \(Q\) consists of the variables in \(D\) together with those
removed by peeling from \(R\setminus D\).
A connected component of \(G_2\) is either entirely contained in
\(\Core_{H_{0,n}}(R\setminus D)\) or entirely contained in \(Q\).
Indeed, if some but not all of its variables belonged to \(Q\), then
some edge of \(G_2\) would have one endpoint in \(Q\) and the other in
the terminal residual.  The degree-two check represented by that edge
would then have exactly one incidence with the terminal residual,
contradicting the stopping-set property.

Let
\(
    \mathcal B_0
    :=
    \{B\in\mathcal C_2:B\cap D\neq\varnothing\}
\) denote the set of components that meet \(D\).
Since the components of \(G_2\) are disjoint,
\(
    |\mathcal B_0|\leq|D|.
\)
Moreover, every component in \(\mathcal B_0\) is contained in \(Q\).
Fix any legal peeling sequence from \(R\setminus D\).
We claim that fewer than \(k\) components of \(G_2\) are contained
in \(Q\).  Suppose otherwise.  Apart from the components in
\(\mathcal B_0\), order the components contained in \(Q\) according
to the first peeling step that removes one of their variables, and let
\(\mathcal B\) consist of \(\mathcal B_0\) together with the earliest
\(k-|\mathcal B_0|\) such components.  Thus
\(
    |\mathcal B|=k.
\)
Consider any
\(B\in\mathcal B\setminus\mathcal B_0\), and let \(f_B\) be the check
that removes the first variable of \(B\).  The check \(f_B\) cannot
have exactly two incidences with \(R\).  Otherwise those two
incidences would define an edge of \(G_2\), whose other endpoint would
also lie in \(B\).  Since \(f_B\) is available for peeling, that other
endpoint would already have been removed, contradicting the choice of
the first removed variable of \(B\).  Since \(R\) is a stopping set,
\(f_B\) therefore has at least three incidences with \(R\).
When \(f_B\) removes the first variable of \(B\), all of its other
variable endpoints in \(R\) have already been removed.  None of them
belongs to \(B\), and by the choice of \(\mathcal B\), their components
belong to \(\mathcal B\).  Hence every incidence of \(f_B\) with \(R\)
has its variable endpoint in
\[
    U:=\bigcup_{B\in\mathcal B}B,
\]
so \(f_B\) is internal to \(U\).  Moreover, the checks \(f_B\) are
distinct for distinct components \(B\), since a check that performs
such a peeling step has no remaining variable afterward.
Thus there are at least
\[
    |\mathcal B|-|\mathcal B_0|
    =
    k-|\mathcal B_0|
    \geq
    k-|D|
    \geq
    k-\left\lfloor\frac{k}{4}\right\rfloor
    =
    \left\lceil\frac{3k}{4}\right\rceil
    =
    t
\]
distinct check nodes internal to \(U\).  This contradicts the fixed
event, under which fewer than \(t\) check nodes are internal to the
union of any \(k\) components.

We have shown that \(Q\) is the union of fewer than \(k\) connected
components of \(G_2\).  If \(Q\neq\varnothing\), let \(j<k\) be the
number of these components.  By \eqref{eq:app-block-union},
\[
    |Q|
    \leq
    \frac{j}{a}\log\frac{M_{\mathrm{comp}}n}{j}.
\]
Because \(\beta<M_{\mathrm{comp}}/e\), the function
\(
    \phi(s):=s\log(M_{\mathrm{comp}}n/s)
\)
is increasing for \(0<s\leq\beta n\), since
\(
    \phi'(s)=\log(M_{\mathrm{comp}}n/s)-1>0.
\)
Therefore
\[
\begin{aligned}
    |Q|
    &\leq
    \frac{k}{a}\log\frac{M_{\mathrm{comp}}n}{k}\\
    &\leq
    \frac{\beta n}{a}\log\frac{2M_{\mathrm{comp}}}{\beta}\\
    &<
    \frac{\rho n}{4}.
\end{aligned}
\]
The same bound is trivial when \(Q=\varnothing\).
Since \(N\geq3\rho n/4\) on the fixed event,
\[
\begin{aligned}
    \left|\Core_{H_{0,n}}(R\setminus D)\right|
    &=
    N-|Q|\\
    &\geq
    \frac{3\rho n}{4}
    -
    \frac{\rho n}{4}
    =
    \frac{\rho n}{2}.
\end{aligned}
\]

Since \(D\) was arbitrary and the high-probability event above was
fixed independently of \(D\), this holds simultaneously for every
\(D\subseteq R\) with \(|D|\leq\beta n/8\).  Hence the theorem holds
with
\[
    \delta_{\mathrm{seed}}=\frac{\beta}{8},
    \qquad
    \eta_{\mathrm{core}}=\frac{\rho}{2}.
\]
\end{proof}

\bibliographystyle{IEEEtranN}
\bibliography{low_weight_dual_closure}

\end{document}